%% file: kr-nested.tex
\documentclass[letterpaper]{article}

\pdfoutput=1

 \usepackage{hyperref}
\usepackage{aaai}
\usepackage{times}
\usepackage{helvet}
\usepackage{courier}
\usepackage{stmaryrd}
\usepackage{url}
\usepackage{amsmath,amsfonts,amssymb}
\usepackage{wasysym}
\usepackage{paralist}
\usepackage{xspace}
\usepackage{boxedminipage}
\usepackage{booktabs}
\usepackage{multirow}
\usepackage[shortlabels]{enumitem}
\usepackage{color}
\usepackage{comment}
\usepackage{subfig}
\usepackage{tikz}

\usepackage{sty-gen-theorem}
\usepackage{macros}

\title{Nested Regular Path Queries in Description Logics}
\author{
 Meghyn Bienvenu\\
 Lab.\ de Recherche en Informatique\\
  CNRS \& Univ.\ Paris Sud, France
 \And
 Diego Calvanese \\
 KRDB Research Centre\\
 Free Univ.\ of Bozen-Bolzano, Italy
 \And
  Magdalena Ortiz \\ {\Large \bf Mantas \v{S}imkus}\\
    Institute of Information Systems\\
  Vienna Univ.\ of Technology, Austria
}

\begin{document}
\maketitle

\begin{abstract}
  Two-way regular path queries (2RPQs) have received increased attention
  recently due to their ability to relate pairs of objects by flexibly
  navigating graph-structured data.  They are present in property paths in
  SPARQL~1.1, the new standard RDF query language, and in the XML query
  language XPath.  In line with XPath, we consider the extension of 2RPQs with
  nesting, which allows one to require that objects along a path satisfy
  complex conditions, in turn expressed through (nested) 2RPQs.  We study the
  computational complexity of answering nested 2RPQs and conjunctions thereof
  (CN2RPQs) in the presence of domain knowledge expressed in description logics
  (DLs).  We establish tight complexity bounds in data and combined complexity
  for a variety of DLs, ranging from lightweight DLs (\dlc,
  \el) 
  up to highly expressive ones.  Interestingly,
  we are able to show that adding nesting to (C)2RPQs does not affect
  worst-case data complexity of query answering for any of the considered DLs.
  However, in the case of lightweight DLs, adding nesting to 2RPQs leads to a
  surprising jump in combined complexity, from \ptime-complete to
  \exptime-complete.
\end{abstract}

\section{Introduction}
\label{sec:introduction}

Both in knowledge representation and in databases, there has been great
interest recently in expressive mechanisms for querying data, while taking into
account complex domain knowledge \cite{CaDL08,GlimmLHS08}.  Description Logics
(DLs) \cite{BCMNP03}, which on the one hand underlie the W3C standard Web
Ontology Language (OWL), and on the other hand are able to capture at the
intensional level conceptual modeling formalisms like UML and ER, are
considered particularly well suited for representing a domain of interest
\cite{BoBr03}.  In DLs, instance data, stored in a so-called ABox, is
constituted by ground facts over unary and binary predicates (\emph{concepts}
and \emph{roles}, respectively), and hence resembles data stored in graph
databases \cite{ConsensM90,BLLW12}.  There is a crucial difference, however,
between answering queries over graph databases and over DL ABoxes. In the
former, the data is 
assumed to be complete, hence query answering amounts to the
standard database task of query evaluation.  In the latter, 
it is typically assumed that the data is incomplete 
and additional domain knowledge is provided by the DL ontology (or TBox).
Hence query answering amounts to the 
more complex task of computing \emph{certain answers}, i.e., those answers that
are obtained
from all databases that both contain the explicit facts in the ABox and
satisfy the TBox constraints.  This difference has driven research in
different directions.

In databases, expressive query languages for querying graph-structured data
have been studied, which
are based on the requirement of relating objects by
flexibly navigating the data.  The main querying mechanism that has been
considered for this purpose is that of one-way and two-way regular path queries
(RPQs and 2RPQs) \cite{CrMW87,CDLV03b}, which are queries returning pairs of
objects related by a \emph{path} whose sequence of edge labels belongs to a
regular language over the (binary) database relations and their inverses.
Conjunctive 2RPQs (C2RPQs) \cite{CDLV00b} are a significant extension of such
queries that add to the navigational ability the possibility of expressing
arbitrary selections, projections, and joins over objects related by 2RPQs, in
line with conjunctive queries (CQs) over relational databases.
Two-way RPQs are
present in the property paths in SPARQL~1.1 \cite{W3Crec-SPARQL-1.1-query}, the new
standard RDF query language, and in the XML query language XPath
\cite{W3Crec-XPath-2.0}.  An additional 
construct 
that is present in XPath 
is the possibility of using \emph{existential test operators}, also known as
\emph{nesting}, to express sophisticated conditions along navigation paths.
When an existential test $\ptest{E}$ is used in a 2RPQ $E'$, there will be
objects along the main navigation path for $E'$ that match positions of $E'$
where $\ptest{E}$ appears; such objects are required to be the origin of a path
conforming to the (nested) 2RPQ~$E$.  It is important to notice that
existential tests in general cannot be captured even by C2RPQs, e.g., when tests
appear within a transitive closure of an RPQ.  Hence, adding nesting effectively
increases the expressive power of 2RPQs and of C2RPQs.

In the DL community,
query answering has been investigated extensively for a wide range of DLs,
with much of the work devoted to CQs. 
%
With regards to the complexity of query answering,
attention has been paid on the one hand to
\emph{combined complexity}, i.e., the complexity measured considering as input
both the query and the DL knowledge base (constituted by TBox and ABox), and on
the other hand to \emph{data complexity}, i.e., when only the ABox is
considered as input. 
For 
expressive DLs that extend \alc,
CQ answering is typically \conp-complete in data-complexity
\cite{OrCE08}, and
\twoexptime-complete in combined complexity
\cite{GlimmLHS08,Lutz08,EiterLOS09}.
For lightweight DLs, instead, CQ answering is
in \acz in data complexity  for
\dlc 
\cite{CDLLR07}, and \ptime-complete
for 
\el 
\cite{kris-Lutz-DL-07}.
For both logics, the combined complexity is dominated by the \np-completeness of CQ evaluation over
plain relational databases.
There has also been some work on
(2)RPQs and C(2)RPQs. 
For the very
expressive DLs
$\mathcal{ZIQ}$,  $\mathcal{ZOQ}$, and $\mathcal{ZOI}$,
where regular expressions over roles are present also in the DL,
a \twoexptime upper bound has been shown  via
techniques based on alternating automata over infinite trees
\cite{calv-etal-09}.
For the Horn fragments of $\mathcal{SHOIQ}$ and  $\mathcal{SROIQ}$,
\ptime-completeness in data complexity and \exptime/2\exptime-completeness in
combined complexity are known \cite{OrtizRS11}.
For lightweight
DLs, tight bounds for answering 2RPQs and C2RPQs have only very recently been established
 by \citeauthor{BienvenuOS13}~\shortcite{BienvenuOS13}: for
(C)(2)RPQs, data complexity is \nlspace-complete in \dlc and \dlh, and
\ptime-complete in \el and \elh.  For all of these logics, combined complexity is
\ptime-complete for (2)RPQs and \pspace-complete for C(2)RPQs.

Motivated by the expressive power of nesting in XPath and SPARQL, in this paper
we significantly advance these latter lines of research on query answering in
DLs, and study the impact of adding nesting to 2RPQs and C2RPQs.
We establish tight complexity 
bounds in data and combined complexity for a variety of DLs, ranging from
lightweight DLs of the \dlc and \el families
up to the highly expressive ones of the \sh and $\mathcal{Z}$ families.
Our results are summarized in Table~\ref{tab:results}. For DLs containing at least
\eli, we are able to encode away nesting, thus showing that the worst-case
complexity of query answering is not affected by this construct.  Instead, for
lightweight DLs (starting already from \dlc!), we show that adding
nesting to 2RPQs leads to a surprising jump in combined complexity, from
\ptime-complete to \exptime-complete. 
We then develop a sophisticated
rewriting-based technique that builds on (but significantly extends) the one
proposed by \citeauthor{BienvenuOS13}~\shortcite{BienvenuOS13}, which we
use to prove that the problem remains in \nlspace for \dlc.
We thus show that  adding nesting to (C)2RPQs does
not affect worst-case data complexity of query answering
for lightweight DLs.

Some proofs have been relegated to the appendix.~

\begin{table*}[t!]
\renewcommand{\arraystretch}{1.1}
\renewcommand\tabcolsep{4pt}
\centering
\begin{tabular}{rllllll}
& \multicolumn{2}{l}{2RPQ} &\multicolumn{2}{l}{C2RPQ} & \multicolumn{2}{l}{N2RPQ / CN2RPQ}  \\[.15cm]
& data \phantom{fd}& combined \phantom{k}& data \phantom{fd}&combined\phantom{k}& data \phantom{fd}&combined\phantom{k} \\[.1cm]
  \cmidrule[.8pt]{2-7}
Graph DBs \& RDFS
&  \nlspace-c &  \nlspace-c & \nlspace-c & \np-c & \nlspace-c & \textbf{\ptime-c}  / \np-c \\  \cmidrule[.2pt]{2-7}
  \dlc & \nlspace-c & \ptime-c &
\nlspace-c
& \pspace-c &  \textbf{\nlspace-c} &  \textbf{\exptime-c} \\   \cmidrule[.2pt]{2-7}
 Horn DLs (e.g., \el, \hornshiq)& \ptime-c & \ptime-c & \ptime-c & \pspace-c & \textbf{\ptime-c} & \textbf{\exptime-c} \\   \cmidrule[.2pt]{2-7}
 Expressive DLs (e.g., \alc, \shiq)& \conp-h & \exptime-c\phantom{fd} & \conp-h & \twoexptime-c\phantom{fd} & \conp-h &
 \textbf{\exptime-c} / \textbf{\twoexptime-c}\phantom{fd} \\   \cmidrule[.8pt]{2-7}
\end{tabular}
\caption{Complexity of query answering. The `c' indicates
  completeness, the  `h' hardness. New results are marked in bold. For
  existing results, refer to
  \protect\cite{BienvenuOS13,PerezAG10,Baeza13,calv-etal-09,OrtizRS11} and
  references therein. \vspace*{-.3cm}
}
\label{tab:results}
\end{table*}

\section{Preliminaries}
\label{sec:preliminaries}



We briefly recall the syntax and semantics of description logics (DLs).
As usual, we assume countably infinite, mutually disjoint sets \cn, \rn, and \inds\
of \emph{concept names}, \emph{role names}, and \emph{individuals}.
We typically use $A$ for concept names, $p$ for role names, and
$a,b$ for individuals.
 An \emph{inverse role} takes the form $p^-$ where $p\in \rn$.
We let $\rni\,{=}\,\rn\cup\{p^-\,{\mid}\,p\,{\in}\,\rn\}$ and denote by $r$ elements of $\rni$.

A DL knowledge base (KB) consists of a \emph{TBox} and an \emph{ABox}, whose forms
depend on the DL in question.
In the DL \elhibot, a TBox is defined as a set of
\emph{(positive) role inclusions}
of the form $r \sqsubseteq r'$ and
\emph{negative role inclusions}
of the form $r \sqcap r' \sqsubseteq \bot$
 with $r,r' \in \rni$, and
 \emph{concept inclusions} of the form
$C \sqsubseteq D$, where $C$ and $D$ are \emph{complex concepts}
formed according to the following syntax:\footnote{We slightly generalize the
usual \elhibot by allowing for negative role inclusions.}
  \[
    C ::= \top ~\mid~ \bot ~\mid~ A ~\mid~ \exists r.  C  ~\mid~ C\sqcap C
  \]
with $A \in \cn$ and $r \in \rni$.

Some of our results refer specifically to the \emph{lightweight DLs}
that we define next.
\elhi is the fragment of \elhibot that has no $\bot$.
\elh and $\mathcal{ELI}$ are obtained by additionally disallowing inverse
roles and role inclusions, respectively.
\dlh is also a fragment of \elhibot, in which
concept inclusions can only take the forms $B_1 \ISA B_2$ and $B_1
\AND B_2 \ISA \bot$, for $B_i$ a concept name or concept of the form
$\exists{r}.\top$ with $r\in\rni$.
\dlc is the fragment of \dlh that disallows (positive and negative) role
inclusions.

An \emph{ABox} is a set of assertions of the form $C(a)$ or $r(a,b)$,
where $C$ is a complex concept, $r \in \rni$, and $a,b \in \inds$.
We use $\ainds(\Amc)$ to refer to the set of individuals in $\Amc$.


\paragraph{Semantics.}
The semantics of DL KBs
is based upon \emph{interpretations}, which take the form
$\Imc = (\Delta^{\Imc}, \cdot^{\Imc})$, where $\Delta^{\Imc}$
is a non-empty set and $\cdot^{\Imc}$ 
maps each $a \in \inds$ to $a^{\Imc} \in \Delta^{\Imc}$,
each $A \in \cn$ to $A^{\Imc} \subseteq \Delta^{\Imc}$, and
each $p \in \rn$ to $p^{\Imc} \subseteq \Delta^{\Imc} \times
\Delta^{\Imc}$.\footnote{Note that we do not make the \emph{unique name
  assumption} (UNA), but all of our results continue to hold if the UNA is
 adopted.}
The function $\cdot^{\Imc}$ can be straightforwardly extended to complex
concepts and roles.
In the case of \elhibot, this is done as follows:
$\top^\Imc=\Delta^\Imc$, $\bot^\Imc= \emptyset$,
$(p^{-})^{\Imc}=\{(c,d) \mid (d,c) \in p^{\Imc}\}$,
$(\exists r. C)^{\Imc}=\{c \mid \exists d:\, (c,d) \in r^{\Imc}, d \in C^{\Imc}\}$,
and $(C \sqcap D)^\Imc = C^\Imc \cap D^\Imc$.
An interpretation  $\Imc$ satisfies an inclusion
$G \sqsubseteq H$ if $G^{\Imc} \subseteq H^{\Imc}$,
and it satisfies an assertion
$C(a)$ (resp., $r(a,b)$) if
$a^{\Imc} \in A^{\Imc}$ (resp., $(a^{\Imc},b^{\Imc})\in r^{\Imc}$).
A \emph{model} of a KB $(\Tmc, \Amc )$ is an interpretation $\Imc$ which
satisfies all inclusions in $\Tmc$ and assertions in~\Amc.

\paragraph{Complexity.}
In addition to P and (co)NP,
our results refer to the complexity classes
\nlspace (non-deterministic logarithmic space),
\pspace (polynomial space), and (2)\exptime ((double) exponential time),
cf.\ \cite{Papa93}.

\section{Nested Regular Path Queries}

We now introduce our query languages.  In RPQs, nested RPQs and their
extensions, \emph{atoms} are given by (nested) regular expressions whose
symbols are \emph{roles}. The set
$\sigmaroles$ of roles contains 
 $\rni$, and all \emph{test
  roles} 
of the forms $\ctest{\{a\}}$
and $\ctest{A}$ with 
$a\in \inds$ and $A\in \cn$. They are interpreted as $\INT{\ctest{\{a\}}}=
(\Int{a},\Int{a})$, and $\INT{\ctest{A}}=\{(o,o) \mid
o\in\Int{A}\}$.

\begin{definition}\label{def:nrpqs-syntax}
   A \emph{nested regular expression}
  (NRE), denoted by $E$, is constructed according to the
  following syntax:
  \[
  E ::= \sigma ~\mid~ E\cdot E ~\mid~ E\cup E ~\mid~ E^* ~\mid~
  \ptest{E}
  \]
  where $\sigma\in\sigmaroles$.

  We assume a countably infinite set \vn  of \emph{variables}
  (disjoint from \cn, \rn, and \inds). Each $t\in \vn\cup \inds$ is  a
  \emph{term}. 
 An \emph{atom} is either a \emph{concept atom} of the form $A(t)$,
  with $A\in \cn$ and $t$ a term, or a \emph{role atom}
  of the form $E(t,t')$, with $E$ an NRE and $t,t'$  two
  (possibly equal) terms.


  A \emph{nested two-way regular path query} (\nq) $q(x,y)$ is an atom
  of the form $E(x,y)$, where $E$ is an NRE
  and $x,y$ are two distinct variables.  A \emph{conjunctive \nq}
  (\cnq) $q(\vec{x})$ with \emph{answer variables} $\vec{x}$ has the
  form $\exists\vec{y}.\varphi$, where $\varphi$ is a conjunction of
  atoms whose variables are among $\vec{x}\cup\vec{y}$.

  A (plain) regular expression (RE) is an NRE that does not have
  subexpressions of the form $\ptest{E}$.  \emph{Two-way regular path
    queries} (2RPQs) and \emph{conjunctive 2RPQs}
(C2RPQs) are defined analogously to \nqs and \cnqs but
  allowing only plain REs in atoms.
\end{definition}

Given an interpretation $\Imc$, the semantics of an NRE $E$ is defined
by induction on its structure:
\begin{align*}
  \INT{E_1\cdot E_2} &= \Int{E_1}\circ\Int{E_2},\\
  \INT{E_1\cup E_2} &= \Int{E_1}\cup\Int{E_2},\\
  \INT{E_1^*} &= (\Int{E_1})^*,\\
  \Int{\ptest{E}} &= \{(o,o) \mid
  \begin{array}[t]{@{}l}
    \text{there is } o'\in\dom \text{ s.t. } (o,o')\in\Int{E}\}.
  \end{array}
\end{align*}

A \emph{match} for 
a C2NRPQ $q(\vec{x})=\exists\vec{y}.\varphi$
in an interpretation $\I$ is a
mapping from the terms in $\varphi$ to $\Delta^{\Imc}$ such that
\begin{inparaenum}[(i)]
\item $\pi(a)=\Int{a}$ for every individual $a$ of $\varphi$,
\item $\pi(x)\in \Int{A} $ for every concept atom $A(x)$ of $\varphi$,
  and
\item $(\pi(x),\pi(y))\in \Int{E} $ for every role atom $E(x,y)$ of $\varphi$.
\end{inparaenum}
Let  $\ans(q,\I)= \{\pi(\vec{x})\mid \pi \mbox{ is a match for }q
 \mbox{ in }\I\}$. An individual tuple $\vec{a}$ with the same arity as $\vec{x}$ is
called a \emph{certain answer} to $q$
over a KB $\tuple{\Tmc,\Amc}$ if  $\Int{(\vec{a})}\in \ans(q,\I)$ for every model $\I$ of
$\tuple{\Tmc,\Amc}$. We use $\ans(q,\tuple{\Tmc,\Amc})$ to denote the
set of all certain answers to $q$ over $\tuple{\Tmc,\Amc}$.
In what follows, by \emph{query answering}, we will mean the problem
of deciding whether $\vec{a} \in \ans(q,\tuple{\Tmc,\Amc})$.

\begin{example}
  We consider an ABox of advisor relationships of PhD
  holders%
  \footnote{Our examples are inspired by the \emph{Mathematics
      Genealogy Project
      (\url{http://genealogy.math.ndsu.nodak.edu/})}.}. We assume an
  $\pf{advisor}$ relation between nodes representing academics.
  There are also nodes for theses, universities, research topics, and
  countries, related in the natural way via roles $\pf{wrote}$,
  $\pf{subm}$(itted), $\pf{topic}$, and $\pf{loc}$(ation).
  We give two queries over this ABox.
  \[
  \small q_1(x,y) =
  (\pf{advisor} \conc
  \ptest{\pf{wrote}\conc\pf{topic}\conc\ctest{\pf{Physics}}})^*~(x,y)
  \]
  Query $q_1$ is an \nq that retrieves pairs of a person $x$
  and an academic ancestor $y$ of $x$
  such that all people on the path from $x$ to $y$ (including $y$
  itself) wrote a thesis in Physics.
  \[
  \small
  \begin{array}[t]{@{}l}
    q_2(x,y,z) =
    \pf{advisor}^-(x,z), ~~
    \pf{advisor}^*(x,w),\\
    \!\!\quad \pf{advisor}^-\conc\ptest{\pf{wrote}
      \conc \ptest{\pf{topic}\conc\ctest{\textit{DBs}}}
      \conc \pf{subm}
      \conc \pf{loc} \conc \ctest{\{\textit{usa}\}}}(y,z),\\
    \!\!\quad \big(\pf{advisor} \conc
    \ptest{\pf{wrote}\conc \ptest{ \pf{topic}
        \conc\ctest{\pf{Logic}} } \conc  \pf{subm}
      \conc \pf{loc} \conc{} \ctest{\textit{EU}}}  \big)^*~(y,w)
  \end{array}
  \]
  Query $q_2$ is a \cnq that looks for triples of individuals $x$,
  $y$, $z$ such that $x$ and $y$ have both supervised $z$, who wrote a
  thesis on Databases and who submitted this thesis to a university in
  the USA.  Moreover, $x$ and $y$ have a common ancestor $w$, and all
  people on the path from $x$ to $w$, including $w$, must have written
  a thesis in Logic and must have submitted this thesis to a
  university in an EU country.
  \qedfull
\end{example}

It will often be more convenient to deal with an
automata-based representation of (C)\nqs, 
which we provide
next.

\begin{definition}\label{def:NFAs}
  A \emph{nested NFA} (\nnfa) has the form
  $(\Abf,s_0,F_0)$ where $\Abf$ is an indexed set
  $\{\alpha_1,\ldots,\alpha_n\}$, where each $\alpha_l\in\Abf$ is
  an automaton of the form $(S,s,\delta,F)$, where $S$ is a set
  of states, $s\in S$ is the initial state, $F\subseteq S$ is the set
  of final states, and
  \[
  \delta \subseteq S \times (\sigmaroles \cup \{
  \begin{array}[t]{@{}l}
    \ptest{j_1,\ldots,j_k} \mid\\
    ~ l<j_i\leq n, \text{ for } i\in\{1,\ldots,k\} \}
    ) \times S
  \end{array}
  \]
  We assume that the sets of states of the automata in $\Abf$ are
  pairwise disjoint, and we require that $\{s_0\}\cup F_0$ are states
  of a single automaton in $\Abf$.  If in each transition
  $(s,\ptest{j_1,\ldots,j_k},s')$ of each automaton in $\Abf$ we have
  $k=1$, then the \nnfa is called \emph{reduced}.
\end{definition}
When convenient notationally, we will 
denote an \nnfa $(\Abf,s_0,F_0)$ by $\Abf_{s_0,F_0}$. Moreover, we will use
$S_i, \delta_i,$ and $F_i$ to refer to the 
states, transition
relation, and 
 final states of~$\alpha_i$.

\begin{definition}\label{def:nnfa-semantics}
  Given an interpretation $\Imc$, we define $\Int{\Abf_{s_0,F_0}}$
  inductively as follows.  Let $\alpha_l$ be the (unique) automaton
  in $\Abf$ such that $\{s_0\}\cup F_0\subseteq S_l$.
  Then $(o,o')\in\Int{\Abf_{s_0,F_0}}$ if there is a sequence
  $s_0o_0s_1\cdots o_{k-1}s_ko_k$, for $k\geq 0$, such that $o_0=o$,
  $o_k=o'$, $s_k\in F_0$, and for $i\in\{1,\ldots,k\}$ there is a
  transition $(s_{i-1},\sigma_i,s_i)\in\delta_l$ such that either
  \begin{itemize}[--]
  \item $\sigma_i\in\sigmaroles$ and $(o_{i-1},o_i)\in \Int{\sigma_i}$, or
  \item $\sigma_i=\ptest{j_1,\ldots,j_k}$ such that, for every
    $m\in\{1,\ldots,k\}$, there exists $o_m\in\dom$ with
    $(o_i,o_m)\in\Int{\Abf_{s',F'}}$, where $s'$ and $F'$ are the
    initial and final states of $\alpha_{j_m}$
    respectively. 
  \end{itemize}
\end{definition}


Note that an \nnfa $\Abf_{s_0,F_0}$ such that there are no transitions
of the form $(s,\ptest{j_1,\ldots,j_k},s')$ in the unique
$\alpha_l$ with $\{s_0\} \cup F_0 \subseteq S_l$
is equivalent to a standard NFA.

For every NRE $E$ one can construct in polynomial time an
\nnfa $\Abf_{s_0,F_0}$ such that $\Int{E}=\Int{\Abf_{s_0,F_0}}$ for
every interpretation $\Imc$. This is an almost immediate consequence
of the correspondence between regular expressions and finite state
automata. Moreover, any \nnfa can be transformed into an equivalent
reduced \nnfa by introducing linearly many additional states. In the
following, unless stated otherwise, we assume all \nnfa{}s are reduced.

\section{Upper Bounds via Reductions}
\label{sec:reductions}

In this section, we derive some upper bounds on the complexity of
answering (C)N2RPQs in different DLs, by means of reductions to other
problems. For simplicity, we assume in the rest of this section that query atoms
do not employ test roles of the form
$\{a\}?$. This
is without loss of generality, since each symbol $\ctest{\{a\}}$ can
be replaced by $\ctest{A_a}$ for a fresh concept name $A_a$, by adding
the ABox assertion $A_a(a)$.

We start by showing that answering \cnqs can be polynomially reduced
to answering non-nested C2RPQs using TBox axioms that employ inverses,
conjunction on the left, and qualified existential restrictions.

\begin{proposition}\label{prop:eli-reduction}
  For each \cnq $q$, one can compute in polynomial time an
  $\mathcal{ELI}$ TBox $\Tmc'$ and C2RPQ $q'$ \mbox{such that}
  $\ans(q,\tuple{\Tmc,\Amc})\,{=}\,\ans(q',\tuple{\Tmc \cup
    \Tmc',\Amc})$ for every KB $\tuple{\Tmc,\Amc}$.
\end{proposition}

\begin{proof}
  Let $q$ be an arbitrary \cnq whose role atoms are given by {\nnfa}s,
  that is, they take the form $\Abf_{s_0,F_0}(x,y)$.

  For each atom $\Abf_{s_0,F_0}(x,y)$ in $q$ and each $\alpha_i \in
  \Abf$, we use a fresh concept name $A_s$ for each state $s \in S_i$,
  and define a TBox $\Tmc_{\alpha_i}$ that contains:
  \begin{itemize}
  \item $\top \sqsubseteq A_f$ for each $f \in F_i$,
  \item $\exists r. A_{s'} \sqsubseteq A_s$ for each $(s,r,s') \in
    \delta_i$ with $r \in \rni$,
  \item $A_{s'} \sqcap A \sqsubseteq A_s$ for each $(s,A?,s') \in
    \delta_i$ with $A \in \cn$, and
  \item $A_{s'} \sqcap A_{s_j} \sqsubseteq A_s$ for each
    $(s,\ptest{j},s') \in \delta_i$, with $s_j$ the initial state of
    $\alpha_j$.
  \end{itemize}
  We denote by $\Tmc_{\Abf}$ the union of all $\Tmc_{\alpha_i}$ with
  $\alpha_i \in \Abf$, and define $\Tmc'$ as the union of
  $\Tmc_{\Abf}$ for all atoms $\Abf_{s_0,F_0}(x,y) \in q$. To obtain
  the query $q'$ we replace each atom $\Abf_{s_0,F_0}(x,y)$ by the
  atom $\alpha'_{i}(x,y)$, where $\alpha_i$ is the unique automaton in
  $\Abf$ with $\{s_0\} \cup F_0 \subseteq S_i$, and $\alpha_i'$ is
  obtained from $\alpha_i$ by replacing each transition of the form
  $(s,\ptest{j},s') \in \delta_i$ with $(s,\ctest{A_{s_j}},s')$, for
  $s_j$ the initial state of $\alpha_j$. Note that each $\alpha_i'$ is
  a standard NFA.
  We show in the appendix that $\ans(q,\tuple{\Tmc,\Amc}) =
  \ans(q',\tuple{\Tmc\cup\Tmc',\Amc})$, for every KB
  $\tuple{\Tmc,\Amc}$.
\end{proof}

It follows that in every DL that contains \eli, answering
\cnqs is no harder than answering plain C2RPQs.  From existing upper
bounds  for C2RPQs 
\cite{calv-etal-09,
  OrtizRS11}, we obtain:

\begin{corollary}\label{corollary:cnqs} Answering \cnqs is:
  \begin{itemize}
  \item in \twoexptime in combined complexity for all DLs contained in
    \shiq, \shoi, \ziq, or \zoi.
  \item in \exptime 
    in combined complexity and \ptime in data complexity for all DLs
    contained in \hornshoiq.
  \end{itemize}
\end{corollary}

We point out that the \twoexptime upper bound for expressive DLs can
also be inferred, without using the reduction above,
from the existing results for answering C2RPQs in
$\mathcal{ZIQ}$ and $\mathcal{ZOI}$ \cite{calv-etal-09}.\footnote{For
  queries that do not contain inverse roles, that is, (1-way) CRPQs,
  the same applies to $\mathcal{ZOQ}$ and its sublogics.}  Indeed,
these DLs support regular role expressions as concept constructors,
and a nested expression $\ptest{E}$ in a query can be replaced by a
concept $\exists{E}.\top$ (or by a fresh concept name $A_{E}$ if the
axiom $\exists{E}.\top \ISA A_{E}$ is added to the TBox).
Hence, in $\mathcal{ZIQ}$ and $\mathcal{ZOI}$, nested
expressions provide no additional expressiveness and \cnqs and C2RPQs
coincide.

The construction used in Proposition~\ref{prop:eli-reduction} also
allows us to reduce the evaluation of a \nq to standard reasoning in
any DL that contains $\mathcal{ELI}$.

\begin{proposition}\label{prop:eli-reduction-oneatom}
  For every \nq $q$ and every pair of individuals $a,b$, one can
  compute in polynomial time an $\mathcal{ELI}$ TBox $\Tmc'$, and a
  pair of assertions $A_b(b)$ and $A_s(a)$ such that $(a,b) \in
  \ans(q,\tuple{\Tmc,\Amc})$ iff $\tuple{\Tmc \cup \Tmc',\Amc \cup
    \{A_b(b)\}} \models A_s(a)$, for every DL KB $\tuple{\Tmc,\Amc}$.
\end{proposition}

From this and existing upper bounds for instance checking in DLs, we
easily obtain:

\begin{corollary}\label{corollary:nqs}
  Answering \nqs is in \exptime in combined complexity for every DL that
  contains \eli and is contained in \shiq, \shoi, \ziq, or \zoi.
\end{corollary}

\section{Lower Bounds}
\label{sec:lower-bounds}

The upper bounds we have stated in Section~\ref{sec:reductions} are
quite general, and in most cases
worst-case optimal.

The \twoexptime\ upper bound 
 stated in the first item of
Corollary~\ref{corollary:cnqs} is optimal already for C2RPQs and
$\mathcal{ALC}$.  Indeed, the \twoexptime hardness proof for conjunctive
queries in $\mathcal{SH}$  by
\citeauthor{EiterLOS09}~(\citeyear{EiterLOS09})
can be adapted to use an $\mathcal{ALC}$ TBox and a C2RPQ.
  Also the \exptime bounds in
Corollaries~\ref{corollary:cnqs} and~\ref{corollary:nqs} are
optimal for all DLs that contain $\mathcal{ELI}$, because 
standard reasoning tasks like satisfiability checking are already \exptime-hard
 in this logic \cite{OWLED-short}.  For the same reasons, the
\ptime bound for data complexity in Corollary~\ref{corollary:cnqs} is
tight for \el and its extensions \cite{CalvaneseGLLR06}.

However, for the lightweight DLs \dlh\ and \el, 
the best combined complexity lower bounds we have are
\nlspace\ (resp., \ptime) for \nqs
and \pspace for \cnqs, inherited from the lower bounds for (C)NRPQs 
\cite{BienvenuOS13}.
This leaves a significant gap with respect to the \exptime\ upper bounds in
Corollaries~\ref{corollary:cnqs} and~\ref{corollary:nqs}.

We show next that these upper bounds are tight. This is the one of the
core technical results of this paper, and probably the most surprising
one: already evaluating one \nq in the presence of a \dlc or \el TBox
is \exptime-hard. 

\begin{theorem}In \dlc\ and \el, N2RPQ answering
   is \exptime-hard in combined
  complexity.
\end{theorem}

\def\blank{\ensuremath{\mathfrak{b}}}
\def\chooseleft{c_L}
\def\chooseright{c_R}

\begin{proof}
  We provide a reduction from the word problem for \emph{Alternating
    Turing Machines (ATMs)} with polynomially bounded space, which is
  known to be \exptime-hard \cite{ChandraKS81}.
An ATM is given as a tuple
  $M=(\Sigma,S_{\exists},S_{\forall},\delta,s_{\mathit{init}},s_{\mathit{acc}},s_{\mathit{rej}})$,
  where $\Sigma$ is an \emph{alphabet}, $S_\exists$ is a set of
  \emph{existential states}, $S_\forall$ is a set of \emph{universal
    states}, $\delta\subseteq (S_{\exists}\cup S_{\forall}) \times
  \Sigma \cup \{\blank\} \times (S_{\exists}\cup S_{\forall})\times \Sigma
  \cup \{\blank\} \times \{-1,0,+1\}$ is a \emph{transition relation}, $\blank$
  is the blank symbol, and
  $s_{\mathit{init}},s_{\mathit{acc}},s_{\mathit{rej}} \in
  S_{\exists}$ are the \emph{initial state}, the \emph{acceptance
    state} and the \emph{rejection state}, respectively.

  Consider a word $w\in \Sigma^*$. We can w.l.o.g.\,assume that
  $\Sigma=\{0,1\}$, that $M$ uses
  only $|w|$ tape cells and that $|w|\geq 1$. 
  Let $m=|w|$, and, for each $1 \leq i \leq m$, let $w(i)$ denote the
  $i$th symbol of $w$.  Let
  $S=S_{\exists}\cup S_{\forall}$. We 
  make the following
  further assumptions: \setlist{leftmargin=.75cm,labelindent=0cm}
  \begin{compactenum}[(i)]
  \item The initial state 
   is not a final state:
    $s_{\mathit{init}} \not\in
    \{s_{\mathit{acc}},s_{\mathit{rej}}\}$.
  \item Before entering
  a state
   $s_{\mathit{acc}}$ or
    $s_{\mathit{rej}}$, $M$ writes $\blank$ in all $m$ tape cells. 
  \item There exist 
  functions $\delta_1, \delta_2 : S \times
    \Sigma \cup \{\blank\} \rightarrow S\times \Sigma \cup \{\blank\} \times
    \{-1,0,+1\}$ such that
    $\{\delta_1(s,\sigma),\delta_2(s,\sigma)\}=\{(s',\sigma',d)\mid
    (s,\sigma,s',\sigma',d)\in \delta\}$ for every $s\in S\setminus
    \{s_{\mathit{acc}},s_{\mathit{rej}}\}$ and $\sigma \in
    \Sigma\cup \{\blank\} $. In other words, non-final states of $M$ give
    rise to exactly two successor configurations
    described by the functions $\delta_1, \delta_2$.
  \end{compactenum}
  Note that the machine $M$ can be modified in polynomial time to ensure
  (i-iii), while preserving the acceptance of $w$. 

We next show how to construct in polynomial
  time a \dlc\ KB $\Kmc=(\Tmc,\Amc)$ 
  and a
  query $q$ such that $M$ accepts $w$ iff
 $a \in \ans(q,\Kmc)$ 
  (we return to \el\
  later). The high-level idea underlying the reduction is to use a KB to
  enforce a tree 
 that contains all possible computations of $M$ on $w$.  
   The query $q$ selects a computation in this tree and verifies that it
  corresponds  to a proper, error-free, accepting run.

  \smallskip\noindent {\it Generating the tree of transitions.} 
 First we construct $\Kmc$, which enforces
a tree whose
  edges correspond to the possible transitions of $M$. More precisely, each edge
  encodes a transition together with the resulting position of the
  read/write head of $M$, and indicates whether the transition is
  given by $\delta_1$ or $\delta_2$. This is implemented using role
  names $r_{p,t,i}$, where $p\in \{1,2\}$,
  $t\in\delta$, and $0 \leq i \leq m+1$.
  To mark the nodes that correspond to the initial (resp., a
  final) configuration of $M$, we employ the concept name
  $A_{\mathit{init}}$ (resp., $A_{\mathit{final}}$), and we 
  use 
  $A_{\exists}$ and $A_{\forall}$ to store the transition type. 

  We let $\Amc = \{ A_{\mathit{init}}(a),A_{\exists}(a)\}$, and then
  we initiate the construction of the tree by including in $\Tmc$
  the axiom
  \begin{equation}
    \label{eq:2}
    A_{\mathit{init}}\ISA \exists r_{p, (s_{\mathit{init}},\sigma,s',\sigma',d),1+d }.
  \end{equation}
  for each $\sigma\in \Sigma\cup \{\blank\}$
  and $p\in \{1,2\}$ such that
  $\delta_p(s_{\mathit{init}},\sigma)=(s',\sigma',d)$.
  To generate further transitions,
 $\Tmc$ contains
  \begin{equation}
    \label{eq:1}
    \exists r^-_{p, (s,\sigma,s',\sigma',d),i} \ISA \exists r_{p' ,(s',\sigma^*,s'',\sigma'',d'),i+d'}
  \end{equation}
  for each $(s,\sigma,s',\sigma',d)\in \delta$, $1 \leq i
  \leq m$, $\sigma^*\in \Sigma\cup \{\blank\}$ and $p,p'\in \{1,2\}$ such
  that $\delta_{p'}(s',\sigma^*)=(s'',\sigma'',d')$.
%
%
  Note that a transition $t'=(s',\sigma^*,s'',\sigma'',d')\in \delta$
  can follow 
  $t=(s,\sigma,s',\sigma',d) \in \delta$ only
  if $\sigma^*$ is the symbol written on 
  tape cell 
  $i$, for $i$ the position of the read/write head after
  executing $t$.  This is not guaranteed by (2). Instead, 
  we ``overestimate'' the possible successive transitions, and use
  the query $q$ to select paths that correspond to a proper
    computation.  

  We complete the definition of $\Tmc$ by adding inclusions to label the
  nodes according to the type of states resulting from
  transitions. For each $1 \leq i \leq m$, $p\in \{1,2\}$ and
  transition $(s,\sigma,s',\sigma',d)\in \delta$, we have the axiom
  \[\exists r^-_{p, (s,\sigma,s',\sigma',d),i} \ISA A_{Q}, \mbox{
    where }\]
  \begin{enumerate}[-]
  \item $A_Q = A_{\mathit{final}}$ if
    $s'\in\{s_{\mathit{acc}},s_{\mathit{rej}}\} $,
  \item $A_{Q}=A_\exists$, if $s'\in
    S_{\exists}\setminus\{s_{\mathit{acc}},s_{\mathit{rej}}\} $,
    and
  \item $A_{Q}=A_\forall$, if $s'\in
    S_{\forall}\setminus\{s_{\mathit{acc}},s_{\mathit{rej}}\} $.
  \end{enumerate}




  \smallskip

  We turn to the construction of the query $q$, for which we employ
  the \nnfa representation.
  We construct an \nnfa $\alpha_q = (\Abf,s,F)$ where
  $\Abf$ has $m+1$ automata
  $\{\alpha_0,\ldots,\alpha_{m}\}$. Intuitively, the automaton
  $\alpha_0$ will be responsible for traversing the tree representing
  candidate computation paths. At nodes corresponding to the end of a
  computation path, $\alpha_0$ launches $\alpha_1,\ldots,\alpha_{m}$
  which ``travel'' back to the root of the tree and test for the
  absence of errors along the way.
 We start by defining the tests 
  $\alpha_1,\ldots,\alpha_{m}$. Afterwards we define $\alpha_0$,
    which 
    selects a set of paths that correspond to a full
    computation,  and launches these tests at the end of each path. 

  \smallskip\noindent {\it Testing the correctness of a computation path.}
   For
  each $1 \leq l \leq m$, the automaton $\alpha_l =
  (S_l,s_l,\delta_l,F_l)$ is built as follows. We let $S_{l} = \{ \sigma_l\mid \sigma\in \Sigma\}\cup
  \{\blank_l\}\cup \{s_l'\}$. That is, $S_{l}$ contains a copy of 
  $\Sigma\cup\{\blank\}$ plus the additional state $s_l'$. We define the
  initial state as $s_l=\blank_l$ and let $F_l=\{s_l'\}$. Finally, the
  transition relation $\delta_l$ contains the following tuples:
  \setlist{leftmargin=*,labelindent=0cm}
  \begin{enumerate}[(T1)]
  \item $(\sigma_l,r_{p,(s,\sigma,s',\sigma',d) ,i }^-,\sigma_l)$ for
    all $1\leq i \leq m$, $p\in \{1,2\}$, all transitions
    $(s,\sigma,s',\sigma',d)\in \delta$, and each $\sigma_l\in S_l
    \setminus \{s_l'\} $ with $l\neq i-d$;
  \item $(\sigma_{l}',r_{p,(s,\sigma,s',\sigma',d),i}^-,\sigma_{l})$
    for all $1\leq i \leq m$, $s\in S$ and $p\in \{1,2\}$ with $
    \delta_p (s,\sigma)=(s',\sigma',d) $ and $l=i-d$;
  \item $(\sigma_l,A_{\mathit{init}}?,s_l')$ for $\sigma={w}(l)$.
  \end{enumerate}
  The working of $\alpha_l$ can be explained as follows. Each state
  $\sigma_l\in S_l \setminus \{s_l'\} $ corresponds to one of the 
  symbols that may be written in position $l$ of the tape during a run
  of $M$. When $\alpha_l$ is launched at some node in a computation tree induced by
  $\Kmc$, it attempts to travel up to the root node, and the only
  reason it may fail is when a wrong symbol
 is written in position $l$ at some point in the computation path.
  Recall that in each final configuration of $M$, all symbols are set
  to the blank symbol, and thus the initial state of $\alpha_l$ is
  $\blank_l$.


  Consider a word $w'\in \sigmaroles^*$ of the form
  \begin{equation}
    r_{p_k,t_k,i_k}^- \cdots r_{p_1,t_1,i_1}^- \cdot A_{\mathit{init}}?\label{eq:3}
  \end{equation}
  that describes a path from some node in the tree induced by $\Kmc$ up to the
  root node $a$.
  We claim that $w'$ is accepted
  by every $\alpha_l$ ($1\leq l \leq m$) just in the case that $t_1, \ldots, t_k$
  is a correct sequence of transitions.
  To see why, first suppose that every $\alpha_l$ accepts $w'$, and let
  $(pos_0,st_0,tape_0)$ be the tuple with $pos_0=1$,
  $st_0=s_{\mathit{init}}$ and $tape_0$ contains for each
  $1\leq l \leq m$, the symbol $\sigma_l$ corresponding to the state
   of $\alpha_l$ when reading $A_{\mathit{init}}$.  Clearly, due to (T3), the tuple
  $(pos_0,st_0,tape_0)$ describes the initial configuration
  of $M$ on input $w$.
  For $1 \leq j \leq k$, if $t_j=(s,\sigma,s',\sigma',d)$, then
  we define $(pos_j,st_j,tape_j)$ as follows: $pos_j=i_j$, $st_j=s'$, and
  $tape_j$ contains for each 
  $1\leq i \leq m$, the state
  of $\alpha_i$ when reading $r_{p_j,t_j,i_j}^-$.
  A simple inductive argument shows that for every $1 \leq j \leq k$,
  the tuple $(pos_j,st_j,tape_j)$ describes the
  configuration of $M$ after applying the transitions $t_1,\ldots,
  t_{j}$ from the initial configuration. Indeed,
  let us assume that $(pos_{j-1},st_{j-1},tape_{j-1})$
  correctly describes the configuration after executing $t_1,\ldots,
  t_{j-1}$ and $t_j=(s,\sigma,s',\sigma',d)$.
  After executing $t_j$, the read/write head
  is in position $pos_{j-1}+d$ and the state is $s'$.
  Since the only way to enforce an $r_{p_j,t_j,i_j}^- $-edge is via axioms
  \eqref{eq:2} and \eqref{eq:1}, we must have $pos_j=pos_{j-1}+d$ and
  $st_j=s'$. It remains to show that $tape_{j}$ describes the tape
  contents after executing $t_j$. Consider some position $1 \leq l \leq  m$.
  There are 
  two cases:
  \setlist{leftmargin=*,labelindent=0cm}
  \begin{enumerate}
  \item $l \neq i_j - d$. In this case, we know that the symbol in
    position $l$ is not modified by executing $t_j$. We have to show
    that $\sigma_l\in tape_{j-1}$ implies $\sigma_l\in tape_{j}$. This
    follows from the construction of $\alpha_l$. In particular, when
    reading ${r_{p_j,t_j,i_j}}^{-}$, it must employ a transition from (T1).
  \item $l = i_j - d$. In this case, 
  after executing $t_j$, we must
    have $\sigma'$ in position $l$. We have to show that $\sigma_l\in
    tape_{j-1}$ implies $\sigma_l'\in tape_{j}$. This again follows
    from the construction of $\alpha_l$. In particular, when reading
    ${r_{p_j,t_j,i_j}}^-$, there is only one possible transition available
    in (T2), namely $(\sigma_l',{r_{p_j,t_j ,i_j }}^-,\sigma_l)$.
  \end{enumerate}
  Conversely, it is easy to see that any word of the form \eqref{eq:3}
  that appears in the tree induced by $\Kmc$ 
  and represents a correct computation
  path will be accepted by all of the $\alpha_l$.

  \smallskip\noindent {\it Selecting a proper computation.}
  It remains to define $\alpha_0$, which selects a
   subtree corresponding to a full candidate computation of $M$, and then launches the
  tests 
  defined above at the end of each path. 
  We let
  $\alpha_0=(S_0,s_0,\delta_0,F_0)$, where
    $S_0=\{s_{\downarrow},\chooseleft,\chooseright,s_{\uparrow},s_{\updownarrow},s_{\mathit{test}},s_f\}$,
$s_0=s_{\downarrow}$,
$F_0=\{s_f\}$, and $\delta_0$ is defined next.

  The automaton operates in two main modes: moving 
  down the tree
  away from 
  the root and moving back up towards the root.
  Depending on the type of the state of $M$, in state $s_{\downarrow}$
  the automaton either selects a child node to process next, or
  chooses to launch the test automata. If the tests are
  successful, it switches to moving up. To this end, $\delta_0$ has
  the following transitions:
  \begin{align*}
    (s_{\downarrow},A_{\exists}?,\chooseleft),
    (s_{\downarrow},A_{\exists}?,\chooseright),
    (s_{\downarrow},A_{\forall}?,\chooseleft), \\
    (s_{\downarrow},A_{\mathit{final}}?,s_{\mathit{test}}), \text{ and
    } (s_{\mathit{test}},\ptest{1,\ldots,m},s_{\uparrow}).
  \end{align*}
  The transitions that implement a step down or up are: 
  \begin{enumerate}[-]
  \item $(\chooseleft,r_{1,t,i},s_{\downarrow})$ for every $1 \leq i \leq m$
    and $t\in \delta$,
  \item $(\chooseright,r_{2,t,i},s_{\downarrow})$ for every $1 \leq i \leq m$
    and $t\in \delta$,
  \item $(s_{\uparrow},r_{1,t,i}^-,s_{\updownarrow})$ for every $1
    \leq i \leq m$ and $t\in \delta$, and
  \item $(s_{\uparrow},r_{2,t,i}^-,s_{\uparrow})$ for every $1
    \leq i \leq m$  and $t\in \delta$.

  \end{enumerate}

  After making a step up from the state $s_{\uparrow}$ via an $r_{1,t,i}^-$-edge, the automaton
  enters the state $s_{\updownarrow}$. Depending on the encountered
  state of $M$, 
  the automaton decides either to verify
  the existence of a computation tree for the alternative transition,
  to keep moving up, or to accept the word.
  This is implemented using the following transitions of $\delta_0$:
  \begin{align*}
    (s_{\updownarrow}, ?A_\forall,\chooseright ), (s_{\updownarrow},
    ?A_\exists,s_{\uparrow} ), \text{ and } (s_{\updownarrow},
    ?A_{\mathit{init}},s_f ).
  \end{align*}

  To conclude the definition of $\alpha_q = (\Abf,s,F)$,  set 
  $s=s_{\downarrow}$ and $F=\{s_f\}$. Note that $\alpha_q$ has a
  constant number of states, so it can be converted into an
  equivalent NRE $E_{q}$ in polynomial time. The desired query is
   $q(x,y)=E_{q}(x,y)$.

  The above \dlc TBox $\Tmc$ can be easily rephrased in \el.
  Indeed, we simply 
  take a fresh concept name $ A_{p,t,i}$ for each role
  $r_{p,t,i}$, 
  and
  replace every axiom $C\ISA \exists
  r_{p,t,i} $ by $C\ISA \exists r_{p,t,i}. A_{p,t,i} $ and 
  every axiom $\exists r_{p,t,i}^- \ISA C$ by $A_{p,t,i} \ISA
  C$.
\end{proof}

The above lower bound for answering N2RPQs hinges on the support for existential concepts
in the right-hand-side of inclusions. If they are
disallowed, then one can find a polynomial time
algorithm \cite{PerezAG10}. However, it was
open until now whether the polynomial-time upper bound is optimal. We
next prove
\ptime-hardness of the problem, already for plain graph databases.



\begin{figure}
  \centering
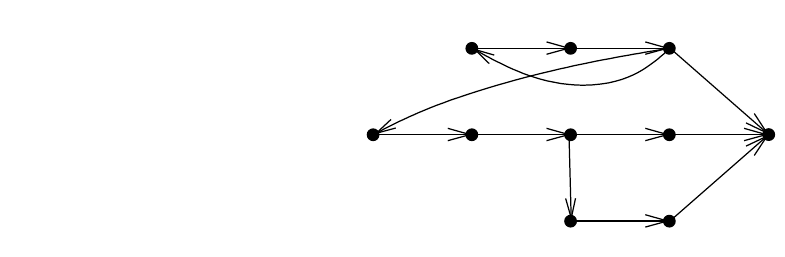  
  \caption{Example ABox in the proof of Theorem~\ref{thm:ptime-lower}}
  \label{fig:abox}
\end{figure}

\begin{theorem}\label{thm:ptime-lower}
  Given as input an N2RPQ $q$, a finite interpretation $\Imc$ and a
  pair $(o,o')\in \Delta^{\Imc}\times \Delta^{\Imc}$, it is
  \ptime-hard to check whether $(o,o')\in \ans(q,\I)$.
\end{theorem}
\begin{proof}
  To simplify the presentation, we prove the lower bound for a slight reformulation
  of the problem. In particular, we show
  \ptime-hardness of deciding $\vec{c}\in
  \ans(q,\tuple{\emptyset,\Amc})$, where $q$ is an N2RPQ and $\Amc$ is
  an ABox with assertions only of the form $A(a)$ or $r(a,b)$, where
  $A\in \cn$ and $r \in \rn$.

  We provide a 
  logspace reduction from the classical \ptime-complete problem of
   checking entailment in
  propositional definite Horn theories. 
  Assume a set
  $T=\{\varphi_1,\ldots,\varphi_n\}$ of definite clauses over a set of propositional variables $V$,
  where each $\varphi_i$ is represented as a rule $v_1 \wedge \ldots \wedge v_m \rightarrow v_{m+1}$.

  Given a variable $g\in V$, we define an ABox $\Amc$, an
  N2RPQ $q$, and 
   tuple $(a_1,a_2)$ such that $T
  \models g$ iff $(a_1,a_2)\in \ans(q,\tuple{\emptyset,\Amc})$.
  We may assume w.l.o.g.\, that $\varphi_1= g\rightarrow g$.
  We define the desired ABox as $\Amc = \Amc_1\cup \Amc_2$,
  using the role names  $s,t$, and 
  $p_v$, where $v\in V$. The ABox
  $\Amc_1$ simply encodes $T$ and contains for
  every 
  $\varphi_i= v_1 \wedge \ldots \wedge v_m \rightarrow v_{m+1}$,
  the following assertions:
\[p_{v_{m+1}}(e^i_{m+1},e^i_{m}),\ldots,p_{v_1}(e_1^{i},e_0^{i}),s(e_0^{i},f). \]
  The ABox $\Amc_2$ 
  links variables in rule bodies 
  with their occurrences in rule heads.
  For every pair of rules
  $\varphi_i= v_1 \wedge \ldots \wedge v_m \rightarrow v_{m+1}$ and
  $\varphi_j= w_1 \wedge \ldots \wedge w_n \rightarrow w_{n+1}$, and
  each $1 \leq l \leq m$ with $v_{l} = w_{n+1} $, it contains the
  assertion $t(e_{l-1}^{i} ,e_{n+1}^{j})$. See Figure~\ref{fig:abox}
  for an example.

  The existence of a proof tree for $g$, which can be
  limited to depth $|V|$, is expressed using the query
  $q(x,y)=E_{|V|}(x,y)$, with $E_1,E_2,\ldots,E_{|V|}$ 
   defined
  inductively: 
  \begin{equation*}
    E_{1}  = \bigcup_{v\in V} \big ( p_v\cdot t \cdot p_v \big )\cdot s\label{eq:4}
  \end{equation*}
  \begin{equation*}
    E_{i}  = \bigcup_{v\in V} \big ( p_v\cdot t \cdot p_v \big )\cdot
    \big ( \tuple{E_{i-1} }\cdot \bigcup_{v\in V} p_v \big )^* \cdot s \qquad (i > 1) \label{eq:5}
  \end{equation*}
  Finally, we let $a_1=e_{1}^1$ and $a_2=f$.
\end{proof}

\section{Concrete Approach for Horn DLs}
\label{sec:horn-dls}

Our complexity results so far leave a gap for the data complexity of
the \dlc family: we inherit \nlspace-hardness from
plain RPQs, 
but we only have the \ptime upper bound
stemming from Proposition~\ref{prop:eli-reduction}.
In this section, we close this gap by providing an \nlspace\ upper
bound. 

This section has an additional
goal.
We recall that the upper bounds in Corollaries~\ref{corollary:cnqs}
and~\ref{corollary:nqs} rely on reductions to answering
(C)2RPQs in extensions of $\mathcal{ELI}$, like
Horn-$\mathcal{SHOIQ}$, $\mathcal{ZIQ}$, and $\mathcal{ZOI}$.
Unfortunately, known algorithms for C2RPQ answering in these logics use
automata-theoretic techniques that are best-case exponential and not
considered suitable for implementation. Hence, we want to provide a direct algorithm that may serve
as a basis for practicable techniques.
To this end, we take an existing algorithm for answering C2RPQs in \elh and \dlh\
due to Bienvenu et al.\ (\citeyear{BienvenuOS13})
and show how it can be extended to handle \cnqs and \elhibot\ KBs.

For presenting the algorithm in this section, it will be useful to
first recall the canonical model property of \elhibot.

\subsection{Canonical Models}
We say that an \elhibot\ TBox $\Tmc$ is in \emph{normal form} if all
of its concept inclusions are of one of the following forms:
\[
A \sqsubseteq \bot \quad A \sqsubseteq \exists r . B \quad \top
\sqsubseteq A \quad B_1 \sqcap B_2 \sqsubseteq A \quad \exists r . B
\sqsubseteq A
\]
with $A,B,B_1,B_2 \in \cn$ and $r \in \rni$.

By introducing fresh concept names to stand for complex concepts,
every TBox $\Tmc$ can be transformed 
in polynomial time into a TBox $\Tmc'$ in normal form that is a model-conservative extension of $\Tmc$. 
Hence, in what follows, we assume that \elhibot TBoxes are in
normal form.

The domain of the canonical model $\Imc_\mathcal{\Tmc, \Amc}$ of a consistent KB
$\tuple{\Tmc, \Amc}$ consists of all sequences $a r_{1} C_{1} \ldots
r_{n} C_{n}$ ($n \geq 0$) such that:
\begin{itemize}
\item $a \in \ainds(\Amc)$ and $r_i\in\rni$ for each $1 \leq i \leq
  n$;
\item each $C_{i}$ is a finite conjunction of concept names;
\item if $n \geq 1$, then $\Tmc,\Amc \models (\exists r_1 .
  C_1)(a)$; 
\item for $1 \leq i < n$, $\Tmc \models C_i \sqsubseteq \exists
  r_{i+1}. C_{i+1}$.
\end{itemize}
\noindent For an
$o\,{\in}\,\Delta^{\Imc_{\Tmc,\Amc}}\,{\setminus}\,\ainds(\Amc)$, we
use $\mathsf{tail}(o)$ to denote its final concept. 
The interpretation $\Imc_{\Tmc,\Amc}$ is then defined as follows:
\begin{align*}
  a^{\Imc_{\Tmc,\Amc}} =    &   \, a \text{ for all } a \in  \ainds(\Amc)  \\
  A^{\Imc_{\Tmc,\Amc}} =   &   \,  \{ a \in \ainds(\Amc) \mid \Tmc,\Amc \models A(a) \} \\
  & \quad \cup \ \{ o \in
  \Delta^{\Imc_{\Tmc,\Amc}}\setminus \ainds(\Amc) \mid
  \Tmc \models \mathsf{tail}(o) \sqsubseteq A \}   \\
  p^{\Imc_{\Tmc,\Amc}} =   &  \,  \{ (a,b) \mid p(a,b) \in \Amc \} \, \cup   \\
  & \, \{ (o_{1},o_{2})
  \mid
  o_{2} = o_{1} r \, C  \text{ and } \Tmc \models r \sqsubseteq p \} \,\cup   \\
  & \, \{ (o_{2},o_{1})
  \mid \ o_{2} = o_{1} r \, C \text{ and } \Tmc \models r \sqsubseteq
  p^{-} \}
\end{align*}
\noindent
Observe that $\Imc_{\Tmc,\Amc}$ is composed of a core part containing
the individuals from $\Amc$ and an \emph{an\-on\-ymous part}
consisting of (possibly infinite) trees rooted at the ABox
individuals.  We use $\Imc_{\Tmc,\Amc}|_{o}$ to denote the restriction
of $\Imc_{\Tmc,\Amc}$ to domain elements having $o$ as a prefix.

It is well-known that the canonical model of a consistent \elhibot\ KB
$\Imc_{\Tmc,\Amc}$ can be homomorphically embedded into any model of
$\tuple{\Tmc,\Amc}$. Since \cnqs\ are preserved under homomorphisms,
we have:

\begin{lemma}\label{lift-canmod}
  For every consistent \elhibot\ KB $\tuple{\Tmc,\Amc}$, \cnq\ $q$, and tuple
  $\vec{a}$ of individuals: $\vec{a} \in \ans(q,\tuple{\Tmc,\Amc})$
  if and only if $\vec{a}\in \ans(q, \canmod)$. 
\end{lemma}

\subsection{Computing Jump and Final 
  Transitions}

A crucial component of our algorithm is to compute relevant partial
paths in a subtree $\canmod|_o$ rooted at an object $o$ in the
anonymous part of \canmod. Importantly,
we also need to remember which parts of the nested automata that 
have been partially navigated below $o$ still need to be
continued.
 This will allow us to `forget' the tree below $o$.

In what follows, it will be convenient use \emph{runs} to talk about
the semantics of {\nnfa}s.
\begin{definition}\label{rundef}
  Let $\Imc$ be an interpretation, and let $(\Abf,s_0,F_0)$ be an
  \nnfa.  Then a \emph{partial run for $\Abf$ 
    on $\Imc$} 
   is a finite node-labelled tree $(T, \ell)$ such that every node is
    labelled
   with an element from $\Delta^\Imc \times (\bigcup_i S_i)$ and for each
  non-leaf node $v$ having label $\ell(v)=(o,s)$ with $s \in
  S_i$, 
  one of the following
  holds: 
  \begin{itemize}
%
  \item $v$ has a unique child $v'$ with $\ell(v')=(o',s')$, and there
    exists
    $(s,\sigma,s') \in \delta_i$ such that $\sigma \in \sigmaroles$ and
    $(o,o') \in \sigma^{\Imc}$;
%
  \item $v$ has exactly two children $v'$ and $v''$ with
    $\ell(v')=(o,s')$ and $\ell(v'')=(o,s'')$, with $s''$ the initial
    state of $\alpha_j$, and there exists a transition $(s, \ptest{j},
    s') \in \delta_i$.
%
  \end{itemize}
  If $T$ has root labelled $(o_1,s_1)$ and a leaf node labelled
  $(o_2,s_2)$ with $s_1,s_2$ states of the same $\alpha_i$,
  then $(T,\ell)$ is called an \emph{$(o_1, s_1, o_2, s_2)$-run}, and
  it is
  \emph{full} if every leaf label $(o',s') \neq (o_2,s_2)$ is
  such that $s' \in F_k$ for some $k$.
\end{definition}

Full runs provide an alternative characterization of the semantics of
{\nnfa}s in Definition~\ref{def:nnfa-semantics}.

\begin{fact}\label{run-fact}
  For every interpretation $\I$, $(o_1,o_2) \in
  \INT{\Abf_{s_1,\{s_2\}}}$ if and only if there is a full $(o_1,s_1,
  o_2, s_2)$-run for $\Abf$ in $\I$.
\end{fact}


We use partial runs to characterize when
an \nnfa $\Abf$
can be partially navigated inside a tree
$\canmod|_o$ whose root satisfies some conjunction of concepts
$C$. Intuitively, $\loops$ stores pairs $s_1$, $s_2$ of states of some $\alpha
\in \Abf$
such that a path from
$s_1$ to $s_2$ exists,
while $\floops$ stores states $s_1$ for which a path to some final
state exists, no matter
where the
final state is reached. Both $\loops$ and $\floops$ store a set $\Gamma$ of
states $s$ of other automata nested in $\alpha$, for which a path from
$s$ to a final state remains to be found.

\begin{definition}\label{jumpdef}
  Let $\Tmc$ be an \elhibot\ TBox in normal form and $(\Abf,s_0,F_0)$
  an \nnfa.  The set $\loops$ consists of tuples $(C,s_1, s_2,
  \pathreq)$ where $C$ is either $\top$ or  a conjunction of concept names from~$\Tmc$,
  $s_1$ and $s_2$ are states from $\alpha_i \in \Abf$, and $\pathreq
  \subseteq \bigcup_{j> i} S_j$.  A tuple $(C,s_1,s_2, \pathreq)$
  belongs to $\loops$ if there exists a partial run $(T,\ell)$ of
  $\Abf$ in the canonical model 
  of $\tuple{\Tmc, \{C(a)\}}$ that satisfies the following conditions:
  \begin{itemize}
  \item the root of $T$ is labelled $(a,s_1)$;
  \item there is a leaf node $v$ 
    with $\ell(v)=(a,s_2)$;
  \item for every leaf node $v$ 
    with $\ell(v)=(o,s)\neq
    (a,s_2)$, 
    either $s\in F_j$ for some $j > i$, or $o=a$ and $s \in \pathreq$.
  \end{itemize}
  The set $\floops$ contains all tuples  $(C, s_1, F, \pathreq)$
  there is a partial run $(T,\ell)$ of $\Abf$ in the canonical
  model 
  of $\tuple{\Tmc, \{C(a)\}}$ that satisfies the following conditions:
  \begin{itemize}
  \item the root of $T$ is labelled $(a,s_1)$;
  \item there is a leaf node $v$ 
    with $\ell(v)=(o,s_f)$ and $s_f \in F$;
  \item for every leaf node $v$ 
    with $\ell(v)=(o,s)$, 
    either $s$ is a final state in some $\alpha_k$, or $o=a$ and $s
    \in \pathreq$.
  \end{itemize}
\end{definition}


\begin{proposition}\label{jump-exptime}
  It can be decided in exponential time if a tuple belongs to $\loops$
  or $\floops$.
\end{proposition}
\begin{proof}[Proof idea]
  We first show how to use TBox reasoning to decide whether $(C,s_1,
  s_2, \pathreq) 
  \in \loops$.
  For every $\alpha_j \in \Abf$, we introduce a fresh concept name
  $A_s$ for each state $s \in
  S_j$. 
  Intuitively, $A_s$ expresses that there is an outgoing path that
  starts in $s$ and reaches a final state.
  If $\{s_1,s_2\} \subseteq S_i$, then
  we 
  add the following inclusions to $\Tmc$:
  \begin{itemize}
  \item $\top \sqsubseteq A_s$, for every $s\in F_j$ with $j> i$;
  \item $\exists r. A_{s'} \sqsubseteq A_s$, whenever $(s,r,s') \in
    \delta_i$ with $r \in \rni$;
  \item $A_{s'} \sqcap B \sqsubseteq A_s$, whenever
    $(s,B?,s')\in\delta_i$;
  \item $A_{s'} \sqcap A_{s''} \sqsubseteq A_s$, whenever
    $(s,\ptest{j},s')\in\delta_i$ and $s''$ is the initial state of
    $\alpha_j$.
  \end{itemize}
  Let $\Tmc'$ be the resulting TBox. In the long version, we show that
  $(C,s_1, s_2, \pathreq) 
  \in \loops$ iff
$$\Tmc' \models (C \sqcap A_{s_2}
\sqcap \bigsqcap_{s \in \pathreq}
A_s) 
\sqsubseteq A_{s_1}.$$
To decide if $(C, s_1, F, \Gamma) \in \floops$, we
must also include in $\Tmc'$ the following inclusions:
\begin{itemize}
\item $\top \sqsubseteq A_s$, for every $s\in F$.
\end{itemize}
We then show that $(C,s_1, F, \pathreq) 
\in \floops$ iff
$$\Tmc' \models (C \sqcap \bigsqcap_{s \in \pathreq} A_s) 
\sqsubseteq A_{s_1}.$$
To conclude the proof, we simply note that both
problems can be decided in single-exponential time, as TBox reasoning
in \elhibot\ is known to be \exptime-complete. %
\end{proof}

\subsection{Query Rewriting}

The core idea of our query answering algorithm 
is to rewrite a given \cnq $q$ into
a set of queries $Q$ such that the answers to $q$ and the union of the
answers for all $q' \in Q$ coincide. However, for evaluating each $q'
\in Q$, we only need to consider mappings from the variables to the
individuals in the core of \canmod.  Roughly, a rewriting step makes
some assumptions about the query variables that are mapped deepest
into the anonymous part and, using the structure of the canonical
model, generates a query whose variables are matched one level closer
to the core.  Note that, even when we assume that no variables are
mapped below some element $o$ in \canmod, the satisfaction of the
regular paths may require to go below $o$ and back up in different
ways. This is handled using jump and final transitions.
The query rewriting algorithm is an adaptation of the algorithm for C2RPQs in
\cite{BienvenuOS13}, to which the reader may refer for
more detailed explanations and examples.

The query rewriting algorithm is presented in
Figure~\ref{fig-firstrewritingstep}.
In the algorithm, we use atoms of
the form $\ptest{A_{s,F}}(x)$, which are semantically equivalent to
$A_{s,F}(x,z)$ for a variable $z$ not occurring anywhere in the query.
This alternative notation will spare us additional variables and make
the complexity arguments simpler.  
To
slightly simplify the notation, we may write $\Abf_{s,s'}$ instead of
$\Abf_{s,\{s'\}}$.

\begin{figure}[t!]
   \scalebox{0.98}{
  \begin{boxedminipage}{8.45cm}
    \noindent \textsc{Procedure} $\rewrite$\\[1mm] 
    Input: \cnq $q$, \elhibot\ TBox $\Tmc$ in normal form
    \setlist{leftmargin=*,labelindent=0cm}
    \begin{enumerate}
    \item Choose either to output $q$ or to continue.
    \item\label{item:chooseVars} Choose a non-empty set $\leaf
      \subseteq \vars(q)$ and $y \in \leaf$. Rename all variables in
      $\leaf$ to $y$.
    \item\label{item:chooseType} Choose a conjunction $C$ of concept
      names from $\Tmc$ such that $\Tmc \models C \sqsubseteq B$
      whenever $B(y)$ is an atom of $q$.  Drop all such atoms from
      $q$.
    \item\label{item:statesStep} For each atom $\at \in q$ of the form
      $\ptest{\Abf_{s_0,F}}(t)$ or $\Abf_{s_0,F}(t,t')$ with $y \in
      \{t,t'\}$:
      \begin{enumerate}
      \item let $\alpha_i \in \Abf$ be the automaton containing $s_0,F$
      \item choose a sequence $s_1, \ldots, s_{n-1}$ of distinct
        states from $S_i$ and some $s_n \in F$
      \item replace $\at$ by the atoms $\Abf_{s_0,s_1}(t,y)$,
        $\Abf_{s_1,s_2}(y,y)$, \ldots, $\Abf_{s_{n-2},s_{n-1}}(y,y)$,
        and
        \begin{itemize}
        \item $\Abf_{s_{n-1},s_n}(y,t')$ if $\at = \Abf_{s_0,F}(t,t')$,
          or
        \item $\ptest{\Abf_{s_{n-1},s_n}}(y)$ if $\at =
          \ptest{\Abf_{s_0,F}}(y)$.
        \end{itemize}
      \end{enumerate}
    \item\label{item:conceptstep} For each atom $\at_j$ of the form
      $\Abf_{s_{j},s_{j+1}}(y,y)$ or $\ptest{\Abf_{s_{j},s_{j+1}}}(y)$
      in $q$, either do
      nothing, 
      or:
      \begin{itemize}
      \item Choose some $(C,s_{j},s_{j+1},\pathreq) \in \loops$ if
        $\at_j = \Abf_{s_{j},s_{j+1}}(y,y)$.
      \item Choose some
      $(C,s_{j},\!\{s_{j+1}\},\pathreq) \in \floops$ if $\at_j =
        \ptest{\Abf_{s_{j},s_{j+1}}}(y)$.
      \item Replace $\at_j$ by $\{
        \ptest{\Abf_{u,F_k}}(y)\,{\mid}\,u\,{\in}\,\pathreq\,{\cap}\,S_k,
        \}$.
      \end{itemize}
    \item\label{item:chooseStep} Choose a conjunction $D$ of concept
      names from $\Tmc$ and $r, r_{1}, r_{2} \in \rni$ such that:
      \begin{enumerate}
      \item $\Tmc \models D \sqsubseteq \exists r. C$,
        $\Tmc \models r \sqsubseteq r_{1}$, and $\Tmc \models r
        \sqsubseteq r_{2}$.
      \item For each atom $\Abf_{u,U}(y,t)$ of $q$ with $u \in S_i$,
        there exists $v \in S_i$ such that $(u, r_1^-, v)\in \delta_i$
      \item For each atom $\Abf_{u,U}(t,y)$ of $q$ with $u \in S_i$,
        there exists $v \in S_i$ and $v' \in U$
        with 
        $(v,r_2, v') \in \delta_i$.
      \item For each atom $\ptest{\Abf_{u,U}}(y)$ of $q$ with $u \in
        S_i$, there exists $v \in S_i$ such that $(u, r_1^-, v)\in
        \delta_i$.
      \end{enumerate}
      For atoms $\Abf_{u,U}(y,y)$, both (b) and (c) apply.
    \item\label{item:makeStep} Replace
      \begin{itemize}
      \item each atom $\Abf_{u,U}(y,t)$ with $t\neq y$ by
        $\Abf_{v,U}(y,t)$,
      \item each atom $\Abf_{u,U}(t,y)$ with $t\neq y$ by
        $\Abf_{u,v}(y,t)$,
      \item each atom $\Abf_{u,U}(y,y)$ by atom $\Abf_{v,v'}(y,y)$,
        and
      \item each atom $\ptest{\Abf_{u,U}}(y)$ by atom
        $\ptest{\Abf_{v,U}}(y)$
      \end{itemize}
      with $v,v'$ as in Step~\ref{item:chooseStep}.
    \item\label{item:addvarStep} Add $A(y)$ to~$q$ for each $A \in D$
      and return to Step 1.
    \end{enumerate}
  \end{boxedminipage}}
  \caption{Query rewriting procedure $\rewrite$.\vspace*{-.4cm}}
  \label{fig-firstrewritingstep}
\end{figure}

The following proposition states the correctness of the rewriting
procedure. Its proof follows the ideas outlined above and can be found in the
appendix of the long version. Slightly abusing notation, we will
also use  $\rewrite(q,\Tmc)$ to denote the set all of queries that can be
obtained by an execution of the rewriting algorithm on $q$ and $\Tmc$.

\begin{proposition}\label{rewrite-correctness}
  Let $\langle \Tmc, \Amc\rangle$ be an \elhibot\ KB and $q(\vec{x})$
  a C2NRPQ.  Then $\vec{a} \in\ans(q, \langle \Tmc, \Amc\rangle)$ iff  there exists $q' \in
  \rewrite(q,\Tmc)$ and a match $\pi$ for $q'$ in $\canmod$ such that $\pi(\vec{x})=\vec{a}$ and $\pi(y) \in
  \ainds(\Amc)$ for every variable $y$ in $q'$.
\end{proposition}


We note that the query rewriting does not introduce fresh terms. Moreover, it
employs an at most quadratic number of linearly sized
\nnfa{}s, obtained 
from the \nnfa{}s of the input  query. Thus, the size of each $q'\in
\rewrite(q,\Tmc)$ is polynomial in the size of $q$ and $\Tmc$. Given
that all the employed checks in Figure~\ref{fig-firstrewritingstep}
can be done in exponential time (see Proposition~\ref{jump-exptime}),
we obtain the following.

\begin{proposition}\label{rewrite-complexity}
  The set $\rewrite(q,\Tmc)$ can be computed in exponential time in
  the size of $q$ and $\Tmc$.
\end{proposition}


\subsection{Query Evaluation}

In Figure \ref{fig-evalatom}, we present an algorithm \evalatom\ for
evaluating N2RPQs.  The idea is similar to the standard
non-deterministic algorithm for deciding reachability: we guess
a sequence $(c_0,s_0) (c_1, s_1) \cdots (c_m, s_m)$ of individual-state pairs,
keeping only two successive elements in memory at any time.
Every element $(c_{i+1}, s_{i+1})$ must be reached from the
preceding element $(c_i,s_i)$ by a single normal, jump, or final
transition.
Moreover, in order to use a jump or final transition, we must ensure that its associated
conditions are satisfied. To decide if the current individual belongs to $C$, we can
employ standard reasoning algorithms, but
to determine whether an outgoing path exists for one of the states in $\Gamma$,
we must make a recursive call to \evalatom.
Importantly, these recursive calls involve ``lower" automata, and
so the depth of recursion is bounded by the number of automata in the N2RPQ
(and so is independent of $\Amc$).
It follows that the whole procedure can be implemented in non-deterministic
logarithmic space in $|\Amc|$, if we discount the concept and role membership tests.
By exploiting known complexity results for instance checking in  \dlh\ and \elhibot,
we obtain:



\def\current{\mathsf{current}} \def\counter{\mathsf{count}}
\def\maxcount{\mathsf{max}} \def\yes{\ensuremath{\mathtt{yes}}}
\def\no{\ensuremath{\mathtt{no}}} \def\anon{\mathsf{anon}}

\begin{figure}[t!]
   \scalebox{0.98}{
  \begin{boxedminipage}{8.45cm}
    \noindent \textsc{Procedure} $\evalatom$\\[1mm]
    Input: \nnfa $(\Abf,s_0,F_0)$, 
    \elhibot\ KB $\Kmc= \langle \Tmc, \Amc \rangle$ in \\
    \hspace*{.8cm} normal form, 
    $(a,b) \in \ainds(\Amc) \times (\ainds(\Amc) \cup \{\anon\})$
    \setlist{leftmargin=*,labelindent=0cm}
    \begin{enumerate}
    \item Let $i$ be such that $s_0 \in S_i$, and set $\maxcount =
      |\Amc| \times |S_i|+1$.
    \item Initialize $\current=(a,s_0)$ and $\counter=0$.
    \item While $\counter < \maxcount$ and $\current \neq (b,s_f)$
      for 
      $s_f \in F_0$
      \begin{enumerate}
      \item Let $\current =
        (c,s)$.  
      \item Guess 
        a pair $(d,s') \in (\ainds(\Amc) \cup \{\anon\})\times S_i$
        such that one of the following holds:
        \setlist{leftmargin=0.5cm}
        \begin{enumerate}
        \item $d \in \ainds(\Amc)$ and there exists 
          $(s, \sigma, s') \in \delta_i$ with $\sigma \in \sigmaroles$ such
          that $(c,d) \in \sigma^{\canmod}$
        \item $d=c$ and $\loops$ contains a tuple $(C,s, s',\pathreq)$
          such that $c \in C^{\canmod}$ and for every $j>i$ and every
          $u \in \pathreq \cap S_j$, 
$$\evalatom((\Abf,u,F_j), \Kmc, (c,\anon)) = \yes$$
\item $d=\anon$, $s' \in F_0$, and $\floops$ contains a tuple
  $(C,s,F_0,\pathreq)$ such that $c \in C^{\canmod}$ and for every $j>i$
  and every $u \in \pathreq \cap S_j$,
$$\evalatom((\Abf,u,F_j), \Kmc, (c,\anon)) = \yes$$
\end{enumerate}
\item Set $\current = (d,s')$ and increment $\counter$.
\end{enumerate}
\item If $\current =(d,s_f)$ for some $s_f \in F_0$, and either $b=d$
  or $b=\anon$, return $\yes$. Else return $\no$.

%
\end{enumerate}
\end{boxedminipage}}
\caption{\nq\ evaluation procedure $\evalatom$.\vspace*{-.2cm}}
\label{fig-evalatom}
\end{figure}

\begin{figure}[t!]
   \scalebox{0.98}{
  \begin{boxedminipage}{8.45cm}
    \noindent \textsc{Procedure} $\evalquery$\\[1mm]
    Input: Boolean \cnq\ $q$, 
    \elhibot\ KB $\Kmc= \langle \Tmc, \Amc \rangle$ in\\
    \hspace*{.85cm} normal form \setlist{leftmargin=*,labelindent=0cm}
    \begin{enumerate}
    \item Test whether $\Kmc$ is satisfiable, output $\yes$ if not.
    \item Set $Q = \rewrite(q,\Tmc)$. Replace all atoms in $Q$ of
      types $C(a)$, $R(a,b)$ by equivalent atoms of type
      $\Abf_{s_0,F_0}(a,b)$.
    \item Guess some $q' \in Q$ and an assignment $\vec{a}$ of
      individuals to the quantified variables $\vec{v}$ in $q'$
      \begin{itemize}
      \item Let $q''$ be obtained by substituting $\vec{a}$ for
        $\vec{v}$.
      \item For every 
        atom $\Abf_{s_0,F_0}(a,b)$ in $q''$
        \begin{itemize}
        \item[] check if $\evalatom((\Abf,s_0,F_0), \Kmc, (a,b))=\yes$
        \end{itemize}
      \item If all checks succeed, return $\yes$.
      \end{itemize}
    \item Return $\no$.
    \end{enumerate}
  \end{boxedminipage}}
  \caption{\cnq\ entailment procedure $\evalquery$.\vspace*{-.3cm}}
  \label{fig-evalquery}
\end{figure}

\begin{proposition}\label{evalatom-props}
  $\evalatom$ is a sound and complete procedure for \nq\ evaluation
  over satisfiable \elhibot\ KBs.  It can be implemented so as to run
  in non-deterministic logarithmic space (resp., polynomial time) in
  the size of the ABox 
  for \dlh\ (resp., \elhibot) KBs.
\end{proposition}

We present in Figure \ref{fig-evalquery} the complete procedure $\evalquery$ for
deciding  \cnq\ entailment.

\begin{theorem}\label{evalquery-props}
  $\evalquery$ is a sound and complete procedure for deciding \cnq\
  entailment over \elhibot\ KBs.  In the case of \dlh\ KBs, it runs in
  non-deterministic logarithmic space in the size of the ABox.
\end{theorem}
\vspace*{-.25cm}
\begin{proof}[Proof idea]
  Soundness, completeness, and termination of \evalquery\ follow
  easily from the corresponding properties of the component procedures
  \rewrite\ and \evalatom\ (Propositions \ref{rewrite-correctness},
  \ref{rewrite-complexity}, and \ref{evalatom-props}).  In \dlh, KB
  satisfiability is known to be \nlspace-complete in data complexity.
  Since the rewriting step is ABox-independent, the 
  size of queries in $\mathcal{Q}$ can be treated as a constant.  It
  follows that the query $q'$ and assignment $\vec{a}$ guessed in Step 3
  can be stored in logarithmic space in $|\Amc|$.
  By Theorem \ref{evalquery-props},
  each call to \evalatom\ runs in non-deterministic logarithmic space.
\end{proof}

\vspace*{-.1cm}

\begin{corollary}
  \cnq\ entailment over \dlh\ knowledge bases is \nlspace-complete in
  data complexity.
\end{corollary}

\section{Conclusions and Future Work}
We have studied the extension of (C)2RPQs with a nesting construct inspired by
XPath, and have characterized the data and combined complexity of answering
nested 2RPQs and C2RPQs for a wide range of DLs.  The only complexity bound we
leave open is whether the \conp lower-bound in data complexity for expressive
DLs is tight; indeed, the automata-theoretic approach used 
to obtain optimal bounds in combined complexity for these logics does not
seem to provide the right tool for tight bounds in data complexity.

In light of the surprising jump from \ptime to \exptime in the combined
complexity of answering \nqs in lightweight DLs,
a relevant research problem
is to identify classes of \nqs that exhibit better computational properties.
We are also interested in exploring whether
the techniques developed in Section~\ref{sec:horn-dls} 
can be extended to deal with additional query constructs, such as existential
``loop-tests'' or forms of role-value maps.
Finally, containment of \nqs has been studied very recently \cite{Reut13},
but only for plain graph databases,
so it would be interesting to investigate containment also in the presence of DL constraints.

\paragraph{Acknowledgments.}
This work has been partially supported by ANR project PAGODA
(ANR-12-JS02-007-01), by the EU IP Project FP7-318338 \emph{Scalable End-user
 Access to Big Data} (Optique), by the FWF projects T515-N23 and P25518-N23,
and by the WWTF project ICT12-015.

\bibliographystyle{aaai}
\fontsize{9.5pt}{10.5pt} \selectfont


\input{inline-bib.bbl}
 \clearpage
 \appendix

 \input{appendix}

\end{document}

%% file: abox.pdf_t
\begin{picture}(0,0)%
\includegraphics{abox.pdf}%
\end{picture}%
\setlength{\unitlength}{3118sp}%
\begingroup\makeatletter\ifx\SetFigFont\undefined%
\gdef\SetFigFont#1#2#3#4#5{%
  \reset@font\fontsize{#1}{#2pt}%
  \fontfamily{#3}\fontseries{#4}\fontshape{#5}%
  \selectfont}%
\fi\endgroup%
\begin{picture}(4760,1546)(1331,-2304)
\put(4726,-941){\makebox(0,0)[lb]{\smash{{\SetFigFont{9}{10.8}{\rmdefault}{\mddefault}{\updefault}{\color[rgb]{0,0,0}$a_1$}%
}}}}
\put(5701,-1936){\makebox(0,0)[lb]{\smash{{\SetFigFont{9}{10.8}{\rmdefault}{\mddefault}{\updefault}{\color[rgb]{0,0,0}$s$}%
}}}}
\put(5701,-1261){\makebox(0,0)[lb]{\smash{{\SetFigFont{9}{10.8}{\rmdefault}{\mddefault}{\updefault}{\color[rgb]{0,0,0}$s$}%
}}}}
\put(4961,-2226){\makebox(0,0)[lb]{\smash{{\SetFigFont{9}{10.8}{\rmdefault}{\mddefault}{\updefault}{\color[rgb]{0,0,0}$p_{v_2}$}%
}}}}
\put(5571,-1531){\makebox(0,0)[lb]{\smash{{\SetFigFont{9}{10.8}{\rmdefault}{\mddefault}{\updefault}{\color[rgb]{0,0,0}$s$}%
}}}}
\put(3676,-1711){\makebox(0,0)[lb]{\smash{{\SetFigFont{9}{10.8}{\rmdefault}{\mddefault}{\updefault}{\color[rgb]{0,0,0}$p_g$}%
}}}}
\put(3976,-1336){\makebox(0,0)[lb]{\smash{{\SetFigFont{9}{10.8}{\rmdefault}{\mddefault}{\updefault}{\color[rgb]{0,0,0}$t$}%
}}}}
\put(5176,-1336){\makebox(0,0)[lb]{\smash{{\SetFigFont{9}{10.8}{\rmdefault}{\mddefault}{\updefault}{\color[rgb]{0,0,0}$t$}%
}}}}
\put(6076,-1636){\makebox(0,0)[lb]{\smash{{\SetFigFont{9}{10.8}{\rmdefault}{\mddefault}{\updefault}{\color[rgb]{0,0,0}$a_2$}%
}}}}
\put(4951,-1711){\makebox(0,0)[lb]{\smash{{\SetFigFont{9}{10.8}{\rmdefault}{\mddefault}{\updefault}{\color[rgb]{0,0,0}$p_{v_1}$}%
}}}}
\put(4351,-1711){\makebox(0,0)[lb]{\smash{{\SetFigFont{9}{10.8}{\rmdefault}{\mddefault}{\updefault}{\color[rgb]{0,0,0}$p_{v_2}$}%
}}}}
\put(4666,-1876){\makebox(0,0)[lb]{\smash{{\SetFigFont{9}{10.8}{\rmdefault}{\mddefault}{\updefault}{\color[rgb]{0,0,0}$t$}%
}}}}
\put(1346,-1631){\makebox(0,0)[lb]{\smash{{\SetFigFont{9}{10.8}{\rmdefault}{\mddefault}{\updefault}{\color[rgb]{0,0,0}$\begin{array}{l@{~:~}rcl}
   \varphi_1   &  g & \rightarrow  & g \\ \varphi_2  &  v_1 \land v_2 &  \rightarrow  & g \\    \varphi_3  &  & \rightarrow  & v_2   \end{array} $}%
}}}}
\put(4351,-961){\makebox(0,0)[lb]{\smash{{\SetFigFont{9}{10.8}{\rmdefault}{\mddefault}{\updefault}{\color[rgb]{0,0,0}$p_{g}$}%
}}}}
\put(4951,-961){\makebox(0,0)[lb]{\smash{{\SetFigFont{9}{10.8}{\rmdefault}{\mddefault}{\updefault}{\color[rgb]{0,0,0}$p_g$}%
}}}}
\end{picture}%

%% file: appendix.tex

\section{Omitted Proofs}

\medskip

  Given an \nnfa\ $(\Abf, s_0, F_0)$, we define the \emph{level} $\level(j)$ of 
  $\alpha_j\in \Abf$ as follows:
  $\level(j)=0$ if the transitions of $\alpha_j$ do not involve any symbol
  $\ptest{k}$, 
  and otherwise $\level(j)= m+1$ where $m$ is the maximum value of
  $\level(k)$ among all 
  $\ptest{k}$ appearing in a transition of $\alpha_j$.

\medskip

\noindent\textbf{Proposition \ref{prop:eli-reduction}}.
\emph{For each \cnq $q$, one can compute
in polynomial time an $\mathcal{ELI}$ TBox $\Tmc'$ and C2RPQ $q'$ \mbox{such
  that}
  $\ans(q,\tuple{\Tmc,\Amc})\,{=}\,\ans(q',\tuple{\Tmc \cup \Tmc',\Amc})$
for every KB $\tuple{\Tmc,\Amc}$.}
\begin{proof}
Let $\Tmc'$ and $q'$ be defined as in the body of the paper.
To complete the argument, we must show  $\ans(q,\tuple{\Tmc,\Amc}) = \ans(q',\tuple{\Tmc \cup
    \Tmc',\Amc})$, for every DL KB $\tuple{\Tmc,\Amc}$.

  To show the first inclusion $\ans(q,\tuple{\Tmc,\Amc}) \subseteq \ans(q',\tuple{\Tmc \cup \Tmc',\Amc})$,
  it suffices to show that
  \begin{equation}\label{eq:first-app-proof}
  \ans(q,\Imc) \subseteq \ans(q',\Imc)
  \end{equation}
  for every model  $\I$ of
  $\tuple{\Tmc \cup \Tmc',\Amc}$.
  Thus, let $\I$ be an arbitrary model of $\tuple{\Tmc \cup \Tmc',\Amc}$.
  To show that Equation \ref{eq:first-app-proof} holds for $\I$,
  we prove that for every index $j$,
  \begin{center}
  $(\dagger)$ if $(o,o') \in \Int{\Abf_{s_j,F_j}}$ for some $o'$, then
  $o \in \Int{A_{s_j}}$.
  \end{center}
  Note that it follows from $(\dagger)$
  that any sequence $s_0o_0s_1\cdots o_{k-1}s_ko_k$
  witnessing  $(o,o')\in\Int{\Abf_{s_0,F_0}}$ also witnesses $(o,o') \in
  \INT{\alpha'_{i}}$, which yields Equation \ref{eq:first-app-proof}.
  We prove $(\dagger)$ by induction on the level of $\alpha_j$.
  For the base case, 
  assume $\level(j)=0$ and $(o,o') \in \Int{\Abf_{s_j,F_j}}$.
  Then there are no transitions of the form $(s,\ptest{j'},s')$ in $\delta_j$,
  so there is a sequence $s_0o_0s_1\cdots o_{k-1}s_ko_k$
  with $o_0=o$, $o_k=o'$, $s_0 = s_j$, and $s_k\in F_j$
  such that for every $i\in\{1,\ldots,k\}$ there is a
  transition $(s_{i-1},\sigma_i,s_i)\in\delta_j$ with
  $\sigma_i\in\sigmaroles$ and $(o_{i-1},o_i)\in \Int{\sigma_i}$.
  Since $\I$ is a model of $\Tmc_{\Abf}$, we must have $o_k \in \Int{A_{s_k}}$.
  Using the axioms in
  $\Tmc_{\Abf}$ and a simple
  inductive argument, we can show that $o_i \in \Int{A_{s_i}}$ for all
  $i\in\{0,\ldots,k\}$,
  since $(o_{i-1},o_i)\in \Int{\sigma_i}$.
  We thus obtain $o \in \Int{A_{s_j}}$.
  The inductive step is similar, but now
  the sequence $s_0o_0s_1\cdots o_{k-1}s_ko_k$
  may have  some $o_{i-1},s_{i-1},o_i,s_i$
  with  $o_{i-1}=o_i$ and $(s_{i-1},\ptest{j'},s_i)$ for some $j' > j$.
  In this case, by definition
  there is $o''\in\Delta^{\canmod}$ such that
  $(o_i,o'')\in\Int{\Abf_{s_{j'},F_{j'}}}$ with $\level(j') < \level(j)$,
  so by the
  induction hypothesis, $o_{i-1} \in \Int{A_{s_j'}}$.
  This, together with the axioms of $\Tmc_{\Abf}$, suffices to
  inductively  show  $o_i \in \Int{A_{s_i}}$ for all $i\in\{0,\ldots,k\}$,
  completing the proof of  $(\dagger)$.

  Next we must show that
  $\ans(q',\tuple{\Tmc \cup \Tmc',\Amc}) \subseteq  \ans(q,\tuple{\Tmc,\Amc})$.
  To this end, for every model $\I$ of $\tuple{\Tmc,\Amc}$, we define its
  extension $\Imc'$ to the vocabulary of $\Tmc'$
  by interpreting the fresh symbols $A_s$ as follows:
  \[
  \Int[\I']{A_s} = \{o \mid \text{ there is } o'\in\Delta^{\canmod} \text{ such that }
  (o,o')\in \Int{\Abf_{s,F_i}} \}
  \]
  where $\alpha_i$ is the unique automaton in $\Abf$ with $s \in S_i$.
  Clearly, $\Imc'$ is a model of $\tuple{\Tmc,\Amc}$, and it is not hard to
  see that it is also a model of $\Tmc'$.
  Indeed, since the definition of $\I'$ directly reflects
  the way the axioms simulate the semantics of the automata,
  a straightforward induction  on the level $\level(j)$ of $\alpha_j$
  establishes that $\I' \models \Tmc_{\alpha_j}$ for every index $j$.
%
%
  Thus, to prove $\ans(q',\tuple{\Tmc \cup \Tmc',\Amc}) \subseteq  \ans(q,\tuple{\Tmc,\Amc})$,
  it suffices to show that for every model $\Imc$ of $\tuple{\Tmc,\Amc}$,
  $$\ans(q',\Imc') \subseteq \ans(q,\Imc)$$
  with $\I'$ the extension of $\I$ as defined above.
  %
  Consider some atom $\alpha_i'(x,y)$ of $q'$ and a pair
  $(o,o') \in \INT{\alpha_i'}$. It is easy to see that the
  sequence $s_0o_0s_1\cdots o_{k-1}s_ko_k$ witnessing
  $(o,o') \in \INT[\I']{\alpha_i'}$ also witnesses
  $(o,o') \in \Int[\I']{\Abf_{s_0,F_0}}$ for the \nnfa\ $\Abf_{s_0,F_0}$ from which
  $\alpha_i'$ was obtained, since the automata only differ in the transitions
  where $\ptest{j}$ was replaced by $\ctest{A_{s_j}}$.
  For every $(o_{i-1},o_i)$  with
  $(o_{i-1},o_i) \in \Int[\I']{\ctest{A_{s_j}}}$,
  we have $o_{i-1} = o_i$ and $o_i \in \Int[\I']{A_{s_j}}$.
  By construction of $\I'$,
  $o_i \in \Int[\I']{A_{s_j}}$ implies  $(o_i,o'')\in
  \Int{\Abf_{s_j,F_j}}$ for some $o''$.
  Hence $\INT[\I']{\alpha_i'} \subseteq \Int[\I]{\Abf_{s_0,F_0}}$ for every
  atom $\alpha_i'(x,y)$ of $q'$, and so
  we obtain $\ans(q',\Imc') \subseteq \ans(q,\Imc)$,
  as desired.
\end{proof}

\medskip

\noindent
\textbf{Proposition \ref{prop:eli-reduction-oneatom}}.
\emph{  For every \nq $q$ and every pair of individuals $a,b$, one can compute
  in polynomial time an $\mathcal{ELI}$ TBox $\Tmc'$, and a pair of assertions
  $A_b(b)$ and $A_s(a)$
  such that $(a,b) \in \ans(q,\tuple{\Tmc,\Amc})$ iff
  $\tuple{\Tmc \cup \Tmc',\Amc \cup \{A_b(b)\}} \models A_s(a)$,
  for every DL KB $\tuple{\Tmc,\Amc}$.}
\begin{proof}
Consider an arbitrary \nq $q(x,y) = E(x,y)$, a
pair of individuals $a,b$, and a knowledge base $\tuple{\Tmc,\Amc}$.
First we observe that, for every interpretation $\I$,
we have $(a,b) \in \Int{q}$ iff  $(a,a) \in \Int{\ptest{E \conc
    \ctest{\{b\}}}}$.
Using this observation, the claim follows almost directly from the proof of
Proposition~\ref{prop:eli-reduction}.
We take the latter query, and eliminate the symbol $\{b\}$.
That is, we add to the ABox the assertion $A_b(b)$, and take the query
$q' = \ptest{E \conc \ctest{A_b}}$.
Let $\Amc' = \Amc \cup \{A_b(b)\}$.
We know that $(a,b) \in \ans(q,\tuple{\Tmc,\Amc})$ iff
$(a,a) \in \ans(q',\tuple{\Tmc,\Amc'})$.

Let $\Abf_{s_0,F_0}$ be the \nnfa representation
of $\ptest{E \conc \ctest{A_b}}$, and let $\Tmc'$ be $\Tmc_{\Abf}$ as
defined in the proof of Proposition~\ref{prop:eli-reduction}, except that
for the ``main" automaton $\alpha_i$ containing states $\{s_0\} \cup F_0$, we use
axioms $\top \sqsubseteq A_f$ with $f \in F_0$ rather than $f \in F_i$
(since we may have $F_0 \neq F_i$).
We show that $(a,a) \in \ans(q',\tuple{\Tmc,\Amc'})$ iff
$\tuple{\Tmc',\Amc'} \models A_{s_0}(a)$.
For the only if direction,
assume
$(a,a) \in \ans(q',\tuple{\Tmc,\Amc'})$ and
take a model $\I$ of $\tuple{\Tmc',\Amc'}$.
Since $\I \models \tuple{\Tmc,\Amc'}$, we have
$(\Int{a},\Int{a}) \in \Int{\Abf_{s_0,F_0}}$, and it
 follows from the proof of Proposition~\ref{prop:eli-reduction} that
$\Int{a} \in \Int{A_{s_0}}$ as desired.
For the other direction, we assume
$\tuple{\Tmc',\Amc'} \models A_{s_0}(a)$ and consider a model $\I$
of $\tuple{\Tmc,\Amc'}$.
We have shown in the proof of Proposition~\ref{prop:eli-reduction}
that $\Imc$ can be extended to a model $\I'$
of $\tuple{\Tmc',\Amc'}$ in
such a way that $\Int[\I']{a} \in \Int[\I']{A_{s_0}}$ implies
$(\Int[\I']{a},o) \in
\Int[\I']{\Abf_{s_0,F_0}}$ for some $o$. Moreover, since
$\Abf_{s_0,F_0}$ is equivalent to an expression
of the form $\ptest{E'}$,
$(\Int[\I']{a},o) \in
\Int[\I']{\Abf_{s_0,F_0}}$ implies $o = \Int[\I']{a}$
and  $(\Int[\I']{a},\Int[\I']{a}) \in
\Int[\I']{\Abf_{s_0,F_0}}$. Hence $(a,a) \in \ans(q',\tuple{\Tmc,\Amc'})$.
\end{proof}

\medskip

For the next proofs, we will utilize the following facts about canonical models:

\begin{fact}\label{canmod1}
Let $\Tmc$ be an \elhibot\ TBox, and let $C,D$ be \elhibot\ concepts.
Then $\Tmc \models D \sqsubseteq E$ iff $a \in D^{\canmod}$ where $\Amc=\{C(a)\}$.
\end{fact}

\begin{fact}\label{canmod2}
Let $\tuple{\Tmc, \Amc}$ be a consistent \elhibot\ KB, and suppose that $o_1$ and $o_2$
are such that $o_1 \in A^{\canmod}$ iff $o_2 \in A^{\canmod}$ for every concept name $A$.
Then the submodels $\canmod|_{o_1}$ and $\canmod|_{o_2}$ are isomorphic.
\end{fact}

\noindent\textbf{Proposition \ref{jump-exptime}.}
\emph{It can be decided in \exptime\ whether a tuple belongs to $\loops$ or $\floops$. }
\begin{proof}
We give the proof only for $\loops$, but the statement for $\floops$ can be proven in the same fashion.
To complete the proof idea presented in the body of the paper, we must show that
$(C,s_1, s_2, \pathreq) 
\in \loops$ iff
\begin{equation}\label{jump-eq}
\Tmc' \models (C \sqcap A_{s_2}
\sqcap \bigsqcap_{s \in \pathreq} A_s) 
\sqsubseteq A_{s_1}
\end{equation}
Recall that the states $s_1$ and $s_2$ belong to $S_i$.
\medskip

For the first direction, suppose that $(C,s_1, s_2, \pathreq)  \in \loops$.
By Definition \ref{jumpdef}, there exists a partial run $(T,\ell)$ of $\Abf$
in the canonical model $\Imc$
of $\Tmc$ and $\Amc=\{C(a)\}$
that satisfies the following conditions:
\begin{itemize}
\item the root of $T$ is labelled $(a,s_1)$
\item there is a leaf node labelled $(a,s_2)$
\item for every leaf node labelled $(o,s)\neq (a,s_2)$,
either $s\in F_j$ for some $j > i$, or $o=a$ and $s \in \pathreq$
\end{itemize}

Let $\Jmc$ be the canonical model of $\Tmc'$ and $\Amc'=\{C(a)\} \cup \{A_{s_2}(a)\}\cup \{A_s(a) \mid s \in \Gamma\}$.
Note that $\Jmc$ is also a model of $\tuple{\Tmc, \Amc}$,
and so by Fact \ref{canmod2}, 
there must exist a homomorphism
$h: \Delta^\Imc \rightarrow \Delta^\Jmc$ from $\Imc$ into $\Jmc$ such that $h(a)=a$.
By Fact \ref{canmod1}, to show Equation \ref{jump-eq}, it suffices to show that $a \in A_{s_1}{^\Jmc}$.
We obtain $a \in A_{s_1}{^\Jmc}$ as a consequence of the following claim and the facts that $h(a)=a$ and the root node is labelled $(a,s_1)$. \\

\noindent\textbf{Claim 1}. For every label $(o,s)$ occurring in $T$:
$h(o) \in A_s{^\Jmc}$. \smallskip\\
\emph{Proof of claim.} The proof is by induction of the co-depth of nodes in $T$. For the base case, suppose that $v$ is a leaf node and $\ell(v)=(o,s)$. If $o=a$, then we must have $s =s_2$ or $s \in \pathreq$, and in both cases, we have $A_s(a) \in \Amc'$. Since $\Jmc$ is a model of $\Amc'$,
we must have $a \in A_s^{\Jmc}$, which yields $h(o) \in A_s{^\Jmc}$ since $h(o)=h(a)=a$.
Next suppose that $o \neq a$. Then we know from above that $s \in F_j$ for some $j >i$.
It follows that  $\Tmc'$ contains the inclusion $\top \sqsubseteq A_s$,
and so $h(o) \in A_s^{\Jmc}$ trivially holds.

For the induction step, consider a non-leaf node $v$ with $\ell(v)=(o,s)$ and $s \in S_k$,
and suppose that the claim has already been shown for the labels of $v$'s children.
By Definition \ref{rundef}, there are two possibilities for $v$.
The first is that $v$ has a unique child $v'$ with $\ell(v')=(o',s')$
and there is a transition $(s,\sigma,s') \in \delta_k$ such that
$\sigma \in \sigmaroles$ and $(o,o') \in \sigma^{\Imc}$.
Since $h$ is a homomorphism,
we have $(h(o),h(o')) \in \sigma^{\Jmc}$, and by applying the induction
hypothesis, we get $h(o') \in A_{s'}^{\Jmc}$.
If $\sigma = r \in \rni$, then $\Tmc'$ contains the inclusion $\exists r. A_{s'} \sqsubseteq A_s$.
Thus, since $\Jmc$ is a model of $\Tmc'$, we must have $h(o) \in A_s^{\Jmc}$.
If instead $\sigma = B?$, then we know that $o=o'$ and $o \in B^{\Imc}$,
and using the homomorphism, we get $h(o) = h(o') \in B^{\Jmc}$.
The TBox $\Tmc'$ contains the axiom $A_{s'} \sqcap B \sqsubseteq A_s$,
and so since $\Jmc$ is a model of $\Tmc'$, we must have $h(o) \in A_s^{\Jmc}$.
The second possibility is that the node $v$ has two children $v'$ and $v''$ with
$\ell(v')=(o,s')$ and $\ell(v'')=(o,s'')$ where $s''$ the initial state of $\alpha_j$
and $(s, \ptest{j}, s') \in \delta_k$. By the induction hypothesis, we must have
$h(o') \in A_{s'}^{\Jmc}$ and $h(o') \in A_{s''}^{\Jmc}$. We also know that $\Tmc'$
contains the inclusion $A_{s'} \sqcap A_{s''} \sqsubseteq A_s$.
Since $\Jmc$ is a model of $\Tmc'$, we obtain $h(o) \in A_s^{\Jmc}$, thereby
completing the proof of the claim.

\medskip

For the other direction, suppose that Equation \ref{jump-eq} holds.
Let $\Imc$ be the canonical model of $\Tmc$ and $\Amc=\{C(a)\}$, and
let $\Imc_0$ be the interpretation with the same domain as $\Imc$
which interprets the concept
and role names from $\langle \Tmc, \Amc \rangle$ exactly as in $\Imc$,
and which additionally interprets each concept name
$A_s$ with $s \in \{s_2 \} \cup \Gamma$ as $\{a\}$.
The interpretation $\Imc_0$ is a model of $\Amc'=\{C(a)\} \cup \{A_{s_2}(a)\}\cup \{A_s(a) \mid s \in \Gamma\}$.
To extend $\Imc_0$ to a model of $\Tmc'$,
we \emph{chase $\Imc_0$ with the inclusions in $\Tmc' \setminus \Tmc$}.
Formally, we consider an infinite sequence of interpretations
$\Imc_0, \Imc_1, \Imc_2, \ldots$ such that $\Imc_{k+1}$ is obtained from
$\Imc_k$ by selecting an inclusion $D \sqsubseteq A \in \Tmc' \setminus \Tmc$
and an $o \in D^{\Imc_k} \setminus A^{\Imc_k}$
and letting $A^{\Imc_{k+1}} = A^{\Imc_k} \cup \{o\}$ (and interpreting all other
concept and roles names as in $\Imc_k$).
Whenever there are multiple objects $o$ satisfying the condition, we choose
an object $o$ having minimal distance from $a$.
Define $\Jmc$ as the interpretation having the same domain as $\Imc$ and the $\Imc_k$
and interpreting concept and role names as follows:
$$A^\Jmc = \bigcup_{k=0}^\infty\,\, A^{\Imc_k} \qquad p^\Jmc = \bigcup_{k=0}^\infty \,\,p^{\Imc_k}$$
It is not hard to see that $\Jmc$ is a model of $\Tmc'$, and since it is also a model of
$\{C(a)\} \cup \{A_{s_2}(a)\}\cup \{A_s(a) \mid s \in \Gamma\}$, we can apply Equation
\ref{jump-eq} to obtain $a \in A_{s_1}^\Jmc$.
Moreover, we can find some finite $m$ such that $a \in A_{s_1}^{\Imc_m}$ and $a \not \in A_{s_1}^{\Imc_{m-1}}$.
We now use the sequence $\Imc_0, \ldots, \Imc_m$ to construct a partial run of $\Abf$ on $\Imc$
that satisfies the conditions of Definition \ref{jumpdef}.\smallskip\\

\noindent\textbf{Claim 2.}
For every $0 \leq k < m$, if $\Imc_{k+1}$ was obtained from $\Imc_k$ by applying an inclusion $D \sqsubseteq A_s$ to $o$, then there is a partial run $(T,\ell)$ of $\Abf$ on $\Imc$ such that:
\begin{itemize}
\item the root node is labelled $(o,s)$
\item every leaf node is either labelled by $(o',s')$ with $s' \in F_j$ for some $j >i$,
or is labelled $(a,s')$ with $s' \in \{s_2\} \cup \Gamma$.
\end{itemize}
\emph{Proof of claim.}
The proof is by induction on $k$. For convenience, throughout we assume that the state $s$ (which corresponds to the concept $A_s$ on the right-hand-side of the inclusion) belongs to $S_g$.  The base case is when $k=0$. There are four cases to consider depending on the shape of the inclusion that was used to obtain $\Imc_{1}$:
\begin{itemize}
\item Case 1: $\Imc_{1}$ was obtained from $\Imc_0$ by applying an inclusion $\top \sqsubseteq A_s$ to $o$,  for some $s\in F_j$ with $j> i$. Then the tree with a single node labelled $(o,s)$ defines a partial run with the required conditions.
\item Case 2: $\Imc_{1}$ was obtained from $\Imc_0$ by applying $\exists r. A_{s'} \sqsubseteq A_s$ to $o$. Then we must have $o \in (\exists r. A_{s'})^{\Imc_0}$, and hence $(o,a) \in r^{\Imc_0}$, since $a$ is the only object that can belong to $A_{s'}^{\Imc_0}$.
Moreover, since $\exists r. A_{s'} \sqsubseteq A_s$ belongs to $\Tmc$, we must have $(s,r,s') \in \delta_g$. 
The desired partial run has a root node labelled $(o,s)$ with a single child labelled $(a,s')$.
\item Case 3: $\Imc_{1}$ was obtained from $\Imc_0$ by applying $A_{s'} \sqcap B \sqsubseteq A_s$ to $o$. Then we must have $o \in (A_{s'} \sqcap B)^{\Imc_0}$, hence $o=a$, $s'\in \{s_2\} \cup \Gamma$, and $a \in B^{\Imc_0} = B^{\Imc}$. It follows that $(a,a) \in B?^\Imc$, and since $A_{s'} \sqcap B \sqsubseteq A_s$ belongs to $\Tmc'$, we must also have $(s,B?,s') \in \delta_g$. 
We can thus use the tree with root node labelled $(a,s)$ and child labelled $(a,s')$.
\item Case 4: $\Imc_{1}$ was obtained from $\Imc_0$ by applying $A_{s'} \sqcap A_{s''} \sqsubseteq A_s$ to $o$. Here again we must have $o=a$, and both $s'$ and $s''$ must belong to
$\{s_2\} \cup \Gamma$. The presence of $A_{s'} \sqcap A_{s''} \sqsubseteq A_s$ in $\Tmc'$ yields $(s,\ptest{j},s') \in \delta_g$.
Thus, we can use the tree with root labelled $(a,s)$ and children labelled $(a,s')$ and $(a, s'')$.
\end{itemize}
For the induction step, suppose that the claim holds whenever $0 \leq k < p$. There are again four possibilities to consider for $\Imc_{p+1}$:
\begin{itemize}
\item Case 1: $\Imc_{p+1}$ was obtained from $\Imc_p$ by applying an inclusion $\top \sqsubseteq A_s$ to $o$,  for some $s\in F_j$ with $j> i$. As above, we can use the tree with a single node labelled $(o,s)$.
\item Case 2: $\Imc_{p+1}$ was obtained from $\Imc_p$ by applying $\exists r. A_{s'} \sqsubseteq A_s$ to $o$. Then we must have $o \in (\exists r. A_{s'})^{\Imc_p}$, and so there is $o'$ with
$o' \in A_{s'}^{\Imc_p}$. Since $\exists r. A_{s'} \sqsubseteq A_s$ belongs to $\Tmc$, we have $(s,r,s') \in \delta_g$.  If $o'=a$ and $s' \in \{s_2\} \cup \Gamma$, then we can use the same argument as in the base case. Else, there must be some earlier stage at which $o'$ was added to $A_{s'}$. We can thus apply the induction hypothesis to find a partial run $(T,\ell)$ that has root labelled $(o', s')$ and is such that every leaf node is either labelled by $(o'',s'')$ with $s'' \in F_j$ for some $j >i$,
or is labelled $(a,s'')$ with $s'' \in \{s_2\} \cup \Gamma$.
It then suffices to add a new root node labelled $(o,s)$ as a parent of the root of $T$ to obtain a partial run satisfying the desired conditions.
\item Case 3: $\Imc_{p+1}$ was obtained from $\Imc_p$ by applying $A_{s'} \sqcap B \sqsubseteq A_s$ to $o$.
Then we must have $o \in A_{s'}^{\Imc_p}$.
If $o=a$ and $s' \in \{s_2\} \cup \Gamma$, then we can proceed as in the base case.
Otherwise, we can apply the induction hypothesis to find a partial run
$(T,\ell)$ which has root labelled $(o, s')$ and is such that every leaf node
is either labelled by $(o',s'')$ with $s'' \in F_j$ for some $j >i$,
or is labelled $(a,s'')$ with $s'' \in \{s_2\} \cup \Gamma$.
We also know that $o \in \Bmc^{\Imc_p} = \Bmc^{\Imc}$,
and that $(s,B?,s') \in \delta_g$ (because of the presence of $A_{s'} \sqcap B \sqsubseteq A_s$ in $\Tmc'$).
It follows that we can construct the desired partial run by adding a
new root node labelled $(o,s)$ as a parent of the root node of $(T, \ell)$.
\item Case 4: $\Imc_{p+1}$ was obtained from $\Imc_p$ by applying $A_{s'} \sqcap A_{s''} \sqsubseteq A_s$ to $o$. Then we must have $(s,\ptest{j},s') \in \delta_g$, and $s''$ must be the initial state of $\alpha_j$.
We also know that $o \in (A_{s'} \sqcap A_{s''})^{\Imc_p}$, and so either we can proceed as in the base case, or
by applying the induction hypothesis, we can find
partial runs $(T_1, \ell_1)$ and $(T_2,\ell_2)$ which have roots labelled $(o,s')$ and $(o,s'')$ respectively, and which satisfy the other requirements. Then we can obtain a partial run with the desired properties by creating a new root node labelled $(o,s)$ whose children are the root nodes of $(T_1, \ell_1)$ and $(T_2,\ell_2)$.
\end{itemize}
This completes the proof of Claim 2. \\

To finish the argument, by taking $k=m-1$, the above claim yields
a partial run $(T,\ell)$ of $\Abf$ on $\Imc$ such that:
\begin{itemize}
\item the root node is labelled $(a,s_1)$
\item every leaf node is either labelled by $(o',s')$ with $s' \in F_j$ for some $j >i$,
or is labelled $(a,s')$ with $s' \in \{s_2\} \cup \Gamma$.
\end{itemize}
It is easy to see that if a non-root node is labelled $(o',s')$ and $s' \in S_i$ (recall that $s_1,s_2 \in S_i$),
then its parent node must be labelled $(o,s)$ for some $s \in S_i$. Since the states in $\Gamma$
do not belong to $S_i$, it must be the case that there is some leaf node labelled $(a,s_2)$.
We have thus found a partial run satisfying all of the conditions of
Definition \ref{jumpdef}, and so we can conclude
that $(C,s_1, s_2, \Gamma) \in \loops$.
\end{proof}


\noindent{\bf Proposition \ref{rewrite-correctness}}.
\emph{ $\Tmc,\Amc \models q$ if and only if there a match $\pi$ for some query $q' \in \rewrite(q,\Tmc)$
in $\Imc_{\Amc,\Tmc}$ such that $\pi(t) \in \ainds(\Amc)$ for every $t \in \terms(q')$. }\\

We split this proof into the two lemmas, a first showing correctness of the procedure $\rewrite$, and a second showing its completeness. In the proofs of these lemmas, it will prove useful to refer to queries that are produced by a single iteration of $\rewrite$. We thus introduce the set $\onestep(q,\Tmc)$ which contains precisely those queries $q'$ for which there is an execution of $\rewrite(q,\Tmc)$ such that $q'$ is output the first time that the procedure returns to Step 1.

In the proof below, we write  $o \in \Int{\ptest{\Abf_{s,F}}}$ to mean that
$(o,o') \in \Int{{\Abf_{s,F}}}$ for some $o' \in \dom$.
We also use a notion analogous to an
$(o_1,s_1,o_2,s_2)$-run, but that is more convenient for automata of the form
$\ptest{\Abf_{s_0,F_0}}$. Let $(T,\ell)$ be a run for $\Abf$ in $\I$.
We call $(T,\ell)$ an \emph{$(o,s, F)$-run} if the root has label $(o,s)$ and some leaf has
label $(o',s_f)$ for some $o' \in \dom$ and some $s_f \in S_i \cap F$, with
$\alpha_i$ the automaton containing the state $s$. Naturally, we have
that
$o\in \Int{\ptest{\Abf_{s_0,F_0}}}$ iff there is a full $(o,s_0,F_0)$-run for
$\Abf$ in $\I$.


\begin{lemma}\label{first-step-correctness}
If $\Tmc,\Amc \models q'$ for some $q' \in \rewrite(q,\Tmc)$,
then $\Tmc,\Amc \models q$.
\end{lemma}
\begin{proof}
It is sufficient to show that if $q' \in \onestep(q,\Tmc)$ and $\Tmc, \Amc
\models q'$, then $\Tmc, \Amc \models q$. Fix a \cnq\ $q$ and a \elhi TBox $\Tmc$.
Let $q' \in \onestep(q,\Tmc)$ be such that $\Tmc, \Amc \models q'$, and let
$\pi$ be a match for $q'$ in $\Imc_{\Tmc,\Amc}$.

Consider the execution of $\rewrite(q,\Tmc)$ which leads to the query $q'$ being output the first time that the procedure returns to~Step 1. Let $\leaf$ be the non-empty subset of $\vars(q)$ which was selected in~Step 2,
let $C$ be the conjunction of concept names selected in Step~3,
and let $D$  and $r,r_{1},r_{2}$ be the conjunction of concept names
 and the roles selected in Step 6.
Because of Step 8, we know that $q'$ contains an atom $A(y)$
for each $A \in D$,
hence $\pi(y) \in D^{\Imc_{\Tmc,\Amc}}$. Since
$\Tmc \models D \sqsubseteq \exists r. C$,
there must exist an $r$-successor $e$ of $\pi(y)$ in $\Imc_{\Tmc, \Amc}$ with $\tail(e)=C$.
We define a mapping $\pi':\terms(q) \rightarrow \Delta^{\Imc_{\Tmc,\Amc}}$ by setting $\pi'(t)=e$ for every $t \in \leaf$ and setting $\pi'(t)=\pi(t)$ for every $t \in \terms(q') \setminus \{y\}$. This mapping is well-defined since every variable in $q$ either belongs to $\leaf$ or appears in $q'$.

We aim to show that $\pi'$ is a match for $q$ in $\Imc_{\Tmc,\Amc}$. To this end, consider some concept atom $B(t) \in q$.
First suppose that $t \in \leaf$. Then we know that the concept $C$ selected in Step 3 is such that $\Tmc \models C \sqsubseteq B$. We then use the fact that since $t \in \leaf$, we have $\pi'(t)=e \in C^{\canmod}$.
If $t \not \in \leaf$,
then $B(t) \in q'$. As $\pi$ is a match for $q'$, we have $\pi(t) \in B^{\canmod}$.
Using $\pi'(t)=\pi(t)$, we get $\pi'(t) \in B^{\canmod}$.

Now consider some atom $\at \in q$ of the form $\Abf_{s,F}(t,t')$ or
$\ptest{\Abf_{s,F}}(t)$.
For $\Abf_{s,F}(t,t')$,
if both $t \not \in \leaf$ and $t' \not \in \leaf$, then it can be verified that $\at \in q'$.
As $\pi$ is a match for $q'$ in $\Imc_{\Tmc,\Amc}$, it
must be the case that $(\pi(t), \pi(t')) \in  \Abf_{s,F}^{\canmod}$.
Since $\pi'(t)=\pi(t)$ and $\pi'(t)=\pi(t)$, the same holds for
$\pi'$.
Similarly for $\ptest{\Abf_{s,F}}(t)$:
if $t \not \in \leaf$, then $\at \in q'$, and
as $\pi$ is a match for $q'$ in $\Imc_{\Tmc,\Amc}$, it
must be the case that $\pi(t) \in \ptest{\Abf_{s,F}}^{\canmod}$.

Next we consider the
more interesting case in which some term from
$\at$ is in $\leaf$.
That is, either $\at = \Abf_{s,F}(t,t')$
with $\{t,t'\} \cap \leaf \neq \emptyset$,
or   $\at = \ptest{\Abf_{s,F}}(t)$
with $t \in \leaf$.
Then in Step~4,
we have a query containing $\Abf_{s,F}(\sigma(t),\sigma(t'))$ or
$\ptest{\Abf_{s,F}}(\sigma(t))$,
where either $\sigma(t)$ or $\sigma(t')$ is $y$,
with $\sigma$ defined as follows: for $t'' \in \{t,t'\}$,
$\sigma(t'')=t''$ for $t'' \not \in \leaf$ and $\sigma(t'')=y$ for $t'' \in
\leaf$. It follows that in
Step~4,  we will
select a sequence $s_1, \ldots s_{n-1}$ of distinct states and a final state
$s_n$
from the
$\alpha_i \in \Abf$ containing  $s$ and $F$,
such that $s_n \in F$.
If $\at = \Abf_{s,F}(\sigma(t),\sigma(t'))$, we replace it by the atoms $\Abf_{s,s_1}(\sigma(t),y),$
$ \Abf_{s_1,s_2}(y,y),$ $ \ldots,$ $\Abf_{s_{n-2},s_{n-1}}(y,y),$
 $\Abf_{s_{n-1},s_n}(y,\sigma(t'))$.
If $\at = \ptest{\Abf_{s,F}}(y)$, we replace it by the atoms $\Abf_{s,s_1}(\sigma(t),y),$
$ \Abf_{s_1,s_2}(y,y),$ $ \ldots,$ $\Abf_{s_{n-2},s_{n-1}}(y,y),$
$\ptest{\Abf_{s_{n-1},s_n}}(y)$.
Let us denote the former set of atoms by $Q_{\Abf}$, and the latter by
$Q_{\ptest{\Abf}}$.   We now establish the following claim:\\

 \noindent\textbf{Claim 1.}
\begin{itemize}\item If $\pi'$ is a match for
  $Q_{\Abf}$ in $\Imc_{\Tmc,\Amc}$, then $\pi'$ is a match for
  $\Abf_{s,F}(t,t')$
  in $\Imc_{\Tmc,\Amc}$.
\item If $\pi'$ is a match for
 $Q_{\ptest{\Abf}}$ in $\Imc_{\Tmc,\Amc}$, then $\pi'$ is a match for
 $\ptest{\Abf_{s,F}}(t)$ in $\Imc_{\Tmc,\Amc}$.
\end{itemize}

\noindent\emph{Proof of claim.}
For the first item, suppose that $\pi'$ is a match for the atoms in $Q_{\Abf}$ in $\Imc_{\Tmc,\Amc}$. Then this means that
$(\pi'(\sigma(t)),\pi'(y)) \in \INT[\canmod]{\Abf_{s,s_1}}$,
$(\pi'(y),\pi'(y)) \in \INT[\canmod]{\Abf_{s_{i},s_{i+1}}}$
for every $1 \leq i < n-1$,
and
$(\pi'(y),\pi'(\sigma(t))) \in \INT[\canmod]{\Abf_{s_{n-1},s_n}}$.
It suffices to compose the sequences witnessing each of these membership statements
to show that 
$(\pi'(\sigma(t)),\pi'(\sigma(t'))) \in \INT[\canmod]{\Abf_{s,s_n}}$,
hence
  $(\pi'(\sigma(t)),\pi'(\sigma(t'))) \in \INT[\canmod]{\Abf_{s,F}}$.
Because of the way
we defined $\pi'$ and $\sigma$,
we have $\pi'(\sigma(t))=\pi'(t)$
and $\pi'(\sigma(t'))=\pi'(t')$. We thus obtain
$(\pi'(t),\pi'(t')) \in \INT[\canmod]{\Abf_{s,F}}$,
i.e.\  $\pi'$ is a match for
 $\Abf_{s,F}(t,t')$ in $\Imc_{\Tmc,\Amc}$.

For $Q_{\ptest{\Abf}}$, the proof is similar.
If $\pi'$ is a match for the atoms in $Q_{\ptest{\Abf}}$
in $\Imc_{\Tmc,\Amc}$,
then $(\pi'(\sigma(t)),\pi'(y)) \in \INT[\canmod]{\Abf_{s,s_1}}$,
$(\pi'(y),\pi'(y)) \in \INT[\canmod]{\Abf_{s_{i},s_{i+1}}}$
for every $1 \leq i < n-1$,
and $(\pi'(y),o) \in \INT[\canmod]{\Abf_{s_{n-1},s_n}}$ for some object $o$.
By putting together all the sequences witnessing each atom, we obtain
$(\pi'(\sigma(t)),o) \in \ptest{\Abf_{s,F}}^{\canmod}$.
Then since $\pi'(\sigma(t))= \pi'(y)=e =\pi'(t)$,
we obtain $\pi'(t)\in \ptest{\Abf_{s,F}}^{\canmod}$.
(\emph{end proof of Claim 1})

\smallskip

Because of Claim 1, to complete the proof that $\pi'$ is a match for $q$ in $\Imc_{\Tmc,\Amc}$, it is sufficient to show the following:\\

\noindent\textbf{Claim 2.} The following hold for all atoms in  $Q_{\Abf}$ or
$Q_{\ptest{\Abf}}$:
\begin{itemize}
\item
$(\pi'(t_1),\pi'(t_2)) \in \Int[\canmod]{\Abf_{s,s'}}$ for atoms of the form
$\Abf_{s,s'}(t_1,t_2)$, and
\item
$\pi'(t_1) \in
\Int[\canmod]{\ptest{\Abf_{s,s'}}}$ for atoms of the form $\ptest{\Abf_{s,s'}}(t_1)$.
\end{itemize}

\noindent\emph{Proof of claim.}
First consider an atom  $\Abf_{s,s'}(t_1,t_2)$ that was not removed in Step~5.
There are three cases depending on which of $t_1$ and $t_2$ equals $y$. We treat each case separately:\\[.15cm]
\noindent \textbf{Case 1}: $t_1=y$ and $t_2\neq y$. It follows that $t_2=\sigma(t_2)$ and so $\pi'(t_2)=\pi(t_2)$.
In Step~7, we will replace  $\Abf_{s,s'}(t_1,t_2)$ with $\Abf_{s'',s'}(t_1,t_2)$ where
$s'' \in S_i$ is such that $(s,r_{1}^{-},s'') \,{\in}\,\delta_i$.
The atom $\Abf_{s'',s'}(t_1,t_2)$ belongs to $q'$,
so we know that it is satisfied by $\pi$.
More precisely, we know that
$(\pi(t_1),\pi(t_2)) \in \Abf_{s'',s'}^{\canmod}$.
Since $e$ is an $r$-successor of $\pi(y)=\pi(t_1)$
in $\Imc_{\Tmc,\Amc}$ and $\Tmc \models r \sqsubseteq r_{1}$,
it follows that $(e,\pi(t_1)) \in \Abf_{s,s''}^{\canmod}$.
Then putting the above together, and using the fact that
$\pi'(t_1)=\pi'(y)=e$, we obtain $(\pi'(t_1),\pi'(t_2)) \in \Int[\canmod]{\Abf_{s,s'}}$. \\[.15cm]
\noindent \textbf{Case 2}: $t_1 \neq y$ and $t_2=y$. It follows that $t_1=\sigma(t_1)$ and so $\pi'(t_1)=\pi(t_1)$.
In Step~7, we will replace  $\Abf_{s,s'}(t_1,t_2)$ with an atom $\Abf_{s,s''}(t_1,t_2)$ where $s'' \in S_i$
is such that  $(s'',r_{2},s') \in \delta$.
The atom $\Abf_{s,s''}(t_1,t_2)$ appears in $q'$, so it must be satisfied by $\pi$.
We thus have $(\pi(t_1), \pi(t_2)) \in \Abf_{s,s''}^{\canmod}$.
We also know that $e$ is an $r$-successor of $\pi(t_2)=\pi(y)$ in $\Imc_{\Tmc,\Amc}$ and that $\Tmc \models r \sqsubseteq r_{2}$.
From this, we can infer that $(\pi(t_2),e) \in \Abf_{s'',s'}^{\canmod}$
By combining the previous assertions and using the fact that $\pi'(t_2)=e$, we can conclude that $(\pi'(t_1),\pi'(t_2)) \in \Int[\canmod]{\Abf_{s,s'}}$.
\\[.15cm]
\noindent\textbf{Case 3}: $t_1=t_2=y$. In Step~7, we will replace  $\Abf_{s,s'}(t_1,t_2)$ with an atom $\Abf_{v,v'}(t_1,t_2)$
where $(s,r_{1}^{-},v) \,{\in}\,\delta$ and $(v,r_{2},s') \in \delta$.
By combining the arguments used in Cases 1 and 2, we find that
 $(\pi'(t_1),\pi(y) \in \Abf_{s,v}^{\canmod}$,
$(\pi(y),\pi(y)) \in \Abf_{v,v'}^{\canmod}$,
and $(\pi(y),\pi'(t_2)) \in  \Abf_{v',s'}^{\canmod}$, which together yield
$(\pi'(t_1),\pi'(t_2)) \in \Int[\canmod]{\Abf_{s,s'}}$.\\[.15cm]
\indent Next we consider the case of an atom  of the form
$\ptest{\Abf_{s,s'}}(y)$ that is not replaced in Step~5.
Such an atom is replaced in Step~7 by
$\ptest{\Abf_{s'',s'}}(y)$, where
$s'' \in S_i$ is such that $(s,r_{1}^{-},s'') \,{\in}\,\delta$.
The atom $\ptest{\Abf_{s'',s'}}(y)$ belongs to $q'$,
so we know that it is satisfied by $\pi$, and hence
$\pi(y) \in \ptest{\Abf_{s'',s'}}^{\canmod}$.
Since $e$ is an $r$-successor of $\pi(y)$
in $\Imc_{\Tmc,\Amc}$ and $\Tmc \models r \sqsubseteq r_{1}$
it follows that $(e, \pi(y) \in \Abf_{s,s''}^{\canmod}$. 
As $\pi'(y)=e$, we obtain $\pi'(y) \in \Abf_{s,s'}^{\canmod}$.\\[.15cm]
\indent We next treat the case of an atom  $\Abf_{s,s'}(t_1,t_2)$
with $t_1=t_2=y$ that is dropped in Step~5 when
 we chose some
$(C,s,s',\pathreq) \in \loops$.
By definition, we know that there is
a partial run $(T,\ell)$ of $\Abf$
in the canonical model of $\tuple{\Tmc, \{C(a)\}}$
such that:
\begin{itemize}
\item  the root of $T$ is labelled $(a,s)$
\item there is a leaf node labelled $(a,s')$
\item for every leaf node $v$ with $\ell(v)=(o,s'')\neq (a,s')$,
either $s'$ is a final state, or
 $o=a$ and $u \in \pathreq$.
\end{itemize}
We  know from above that $\tail(\pi'(y))=C$ hence $\pi'(y) \in
\Int[\canmod]{C}$.
By Fact \ref{canmod2}, the canonical model of $\tuple{\Tmc, \{C(a)\}}$
is isomorphic to
 $\canmod|_{\pi'(y)}$. It follows that there is a partial run $(T',\ell')$
of $\Abf$ in  $\canmod|_{\pi'(y)}$ satisfying the same conditions,
except with $a$ replaced by $\pi'(y)$.
In Step 5, we added to $q'$ every atom
$\ptest{\Abf_{u,F_k}}(y)$ such that
 $u \in \pathreq \cap S_k$. We can thus use exactly
 the same argument as in the previous case to show that
 for every  $u \in \pathreq \cap S_k$,
 $\pi'(y) \in \ptest{\Abf_{u,F_k}}^{\canmod}$, and hence
 there is a full $(\pi'(y),u,F_k)$-run $(T_u,\ell_u)$ of $\Abf$ on $\canmod$.
By attaching to $(T',\ell')$ the tree $(T_u,\ell_u)$ at each non-final leaf
$v$ with label $(\pi'(y),u)$, we obtain a full
$(\pi'(y),s,\pi'(y),s')$-run for $\Abf$ in $\canmod$.
This shows that
$(\pi'(y),\pi'(y)) \in \INT[\canmod]{\Abf_{s,s'}}$
which yields the desired result given that $t_1=t_2=y$.\\[.15cm]
\indent The proof is analogous for an atom $\ptest{\Abf_{s,s'}}(y)$ in
$Q'_{\Abf}$ that is dropped in Step~5 after choosing some $(C,s,\{s'\},\pathreq) \in
\floops$.
 Again, we know that there is
a partial run of $\Abf$
in the canonical model of $\tuple{\Tmc, \{C(a)\}}$ as in
Definition~\ref{jumpdef},
and as $\pi'(y) \in
\Int[\canmod]{C}$, we can find a corresponding partial run $(T',\ell')$
in  $\canmod|_{\pi'(y)}$.
We know that the root of $T'$ is labelled $(\pi'(y),s)$,
there is a leaf labelled  $(o,s')$,
and for every leaf node $v$
 with $\ell(v)=(o',u)\neq (o,s')$,
and  $u$ not a final state, we have
 $o=\pi'(y)$ and $u \in \pathreq$.
Since the atoms  $\ptest{\Abf_{u,F_k}}(y)$ are added in Step 5,
we can use the same arguments as above to find
for each such $u \in \pathreq \cap S_k$,
a full $(\pi'(y),u,F_k)$-run $(T_u,\ell_u)$ for $\Abf$ in \canmod.
By attaching to $(T',\ell')$ the tree $(T_u,\ell_u)$ at each non-final leaf
$v$ with label $(\pi'(y), u)$, we obtain a full
$(\pi'(y),s, \{s'\})$-run for $\Abf$ in $\canmod$,
which shows that $\pi'(y) \in \ptest{\Abf_{s,s'}}^{\canmod}$.

This concludes the proof of Claim~2, and hence of Lemma~\ref{first-step-correctness}.
\end{proof}


\begin{lemma}
Suppose that $\Tmc,\Amc \models q$ and $\pi$ is a match for $q$ in $\Imc_{\Tmc,\Amc}$ such that $\pi(y)= p \cdot r_0\, C$ and there is no $z \in \vars(q)$ such that $\pi(y)$ is a proper prefix of $\pi(z)$. Then there is a match $\pi'$ for some $q' \in \onestep(q,\Tmc)$ such that:
\begin{itemize}
\item $\pi'(t)=\pi(t)$ for every $t \in \terms(q)$ is such that $\pi(t) \neq
  \pi(y)$, and
\item $\pi'(t)= p$ for every $t \in \terms(q)$ with $\pi(t)=\pi(y)$.
\end{itemize}
\end{lemma}
\begin{proof}
Let $\pi$ be a match for a \cnq\ $q$ in $\Imc_{\Tmc,\Amc}$ and $t$ be such that
$\pi(y)= p \cdot r_0\, C$ and there is no $z \in \vars(q)$ with $\pi(y)$ a proper prefix of $\pi(z)$.
We show how to obtain a query $q'\in \onestep(q,\Tmc)$ and match $\pi'$ with the required properties.
In Step 1 of $\rewrite$, we choose to continue on to Step 2,
where we set $\leaf = \{t \in \vars(q) \mid \pi(t)=\pi(y)\}$.
 We define a function $\sigma$ as follows:
$\sigma(t)=y$ if $t \in \leaf$, and
 $\sigma(t)=t$ otherwise.
At the end of Step~2, we have the query $q_2$ that contains the atoms
$\Abf_{s,F}(\sigma(t),\sigma(t'))$ for each  $\Abf_{s,F}(t,t') \in q$,
$B(\sigma(t))$  for each $B(t) \in q$,
and $\ptest{\Abf_{s,F}}(\sigma(t))$  for each $\ptest{\Abf_{s,F}}(t) \in q$.
In Step~3, we choose the concept $C$.
Consider some atom $B(y) \in q_2$.
We know that there must be some atom
$B(t)\in q$ with $t \in \leaf$.
Since $\pi$ is a match for $q$,  we
must have that $\pi(t) = \pi(y) \in B^{\Imc_{\Tmc,\Amc}}$.
Since $\pi(y)= p \cdot r_0\, C$, it follows from the definition of
canonical models that $\Tmc \models C \sqsubseteq B$,
as required by Step~3.

Next we show how to select a decomposition of atoms in Step~4.
First consider an atom $\Abf_{s_0,F}(t_1,t_2)$
which is present in the query at the start of Step~4,
such that $y \in \{t_1,t_2\}$, 
and let $\alpha_\ell$ be the automaton in $\Abf$
that contains $s_0$ and $F$.
Then we know from above that
there is an atom $\Abf_{s_0,F}(u,u')$ 
in the original query $q$
such that $t_1=\sigma(u)$ and $t_2=\sigma(u')$. 
Since $\pi$ is a match for $q$ with $\pi(t_1)=\pi(u)$ and $\pi(t_2)=\pi(u')$,
we know that $(\pi(t_1),\pi(t_2)) \in \Abf_{s,F}^{\canmod}$.
This must be witnessed by some
sequence $w = u_0o_0u_1\cdots o_{g-1}u_go_g$ of states of $\alpha_\ell$ and
domain objects as in Definition~\ref{def:nnfa-semantics}, which has
$o_0=\pi(t_1)$, $u_0 = s_0$, $o_g=\pi(t_2)$ and $u_g\in F$.
We can assume without loss of generality that $g$ is minimal, i.e.
we cannot find a sequence $s'_0o'_0\cdots s'_{g'}o'_{g'}$
witnessing $(\pi(t_1),\pi(t_2) )\in \Abf_{s_0,F}^{\canmod}$ with $g' <
g$. This ensures that in $w$ there do not exist distinct $i,j$
with $u_i  = u_j$, $o_i = o_j$.
Now let $j_0=0$, $j_n=g$, and
$j_{1}< \ldots < j_{n-1}$ be all of the indices $1 \leq i < g$ such that
$o_{i}=\pi(y)$.
Then define the sequence of states $s_1,\ldots,s_n$ by taking
$s_i = u_{j_i}$ for $1 \leq i < n$, and $s_n = u_g$.
Note that the states $s_1,\ldots,s_{n-1}$ are all pairwise distinct.
%
We can thus
choose this sequence of states in Step~4, and replace the atom
$\Abf_{s_0,F}(t_1,t_2)=\Abf_{s_0,F}(\sigma(u),\sigma(u'))$ with
the atoms:
$\Abf_{s_0,s_1}(\sigma(u),y)$, $\Abf_{s_1,s_2}(y,y),\ldots,$ $ \Abf_{s_{n-2},s_{n-1}}(y,y)$, $\Abf_{s_{n-1},s_n}(y,\sigma(u'))$.
Note that, by the way we selected the states $s_i$, we know that $(\pi(t_1),\pi(y))
\in \Abf_{s,s_1}^{\canmod}$,
$(\pi(y),\pi(y))
\in \Abf_{s_i,s_{i+1}}^{\canmod}$ for $1 \leq i < n-1$, and
$(\pi(y),\pi(t_2))
\in \Abf_{s_{n-1},s_{n}}^{\canmod}$.
That is, $\pi$ is a match for all of
these atoms.
Moreover, for each of the atoms $\Abf_{s_i,s_{i+1}}(y,y)$, 
the sequence
 witnessing $(\pi(y),\pi(y))
\in \Abf_{s_i,s_{i+1}}^{\canmod}$ is either entirely inside or entirely
outside $\canmod|_{\pi(y)}$.

For an atom $\ptest{\Abf_{s_0,F}}(y)$, the decomposition
proceeds analogously.
Let $\alpha_\ell$ be the automaton in $\Abf$
that contains $s_0$ and $F$.
Then we know from above that
there is an atom
$\ptest{\Abf_{s_0,F}}(u)$
in the original query $q$
such that $\sigma(u)=y$.
Since $\pi$ is a match for $q$ with $\pi(y)=\pi(u)$,
we know that $\pi(y) \in \ptest{\Abf_{s_0,F}}^{\canmod}$,
and we can take
a sequence $w = u_0o_0u_1\cdots o_{g-1}u_go_g$ of states of $\alpha_\ell$ and
domain objects that witnesses $\pi(y) \in \ptest{\Abf_{s_0,F}}^{\canmod}$.
Note that we must have $u_0=s_0$, $o_0=\pi(y)$, and $u_g \in F$.
We can suppose without loss of generality that $g$ is minimal, and so there
are no indices $i <j$ such that $u_i=u_j$ and $o_i=o_j$.
We let $j_{1}< \ldots < j_{n-1}$ be all of the indices $1 \leq i < g$ such that
$o_{i}=\pi(y)$, and for Step~4 use the sequence of 
states $s_1,\ldots,s_{n}$ with
$s_i = u_{j_i}$ for $1 \leq i < n$, and $s_n = u_g$.
Similarly to above, we know that, by construction,
$(\pi(y),\pi(y)) \in \Abf_{s_0,s_1}^{\canmod}$,
$(\pi(y),\pi(y))
\in \Abf_{s_i,s_{i+1}}^{\canmod}$ for $1 \leq i < n-1$, and
$\pi(y) \in \ptest{\Abf_{s_{n-1},s_{n}}}^{\canmod}$.
We also know that, for each of these new atoms $\Abf_{s_i,s_{i+1}}(y,y)$ or
$\ptest{\Abf_{s_i,s_{i+1}}}(y)$,
the sequence 
as in Definition~\ref{def:nnfa-semantics} that witnesses $(\pi(y),\pi(y))
\in \Abf_{s_i,s_{i+1}}^{\canmod}$ or $\pi(y)
\in \ptest{\Abf_{s_i,s_{i+1}}}^{\canmod}$ is either
 entirely inside or entirely
outside $\canmod|_{\pi(y)}$.

In Step~5, we will remove all atoms
$\Abf_{s_i,s_{i+1}}(y,y)$ or $\ptest{\Abf_{s_i,s_{i+1}}}(y)$
for which the witnessing sequence lies entirely inside $\canmod|_{\pi(y)}$,
possibly introducing some new atoms of the form $\ptest{\Abf_{u,F_k}}(y)$.
First consider an atom $\at = \Abf_{s_i,s_{i+1}}(y,y)$ 
such that the witnessing sequence $u_{j_i}o_{j_i} \ldots u_{j_{i+1}} o_{j_{i+1}}$
only involves objects from $\canmod|_{\pi(y)}$ (recall that $s_i=u_{j_i}$,
$s_{i+1}= u_{j_{i+1}}$, and $o_{j_i}=o_{j_{i+1}}=\pi(y)$).
Then there must exist
a full
$(\pi(y),s_i,\pi(y),s_{i+1})$-run $(T,\ell)$ for $\Abf$ in $\canmod$
whose ``main path" has the sequence of labels $(\pi(y), s_i)=(o_{j_i}, u_{j_i}),
(o_{j_i +1}, u_{j_i+1}), \ldots, (o_{j_{i+1}}, u_{j_{i+1}}) =(\pi(y), s_{i+1})$.
Take this $(T,\ell)$, and in every branch, look for the first occurrence of a node
$v_p$ labeled $(o',u')$ with $o' \not\in \dom[\canmod|_{\pi(y)}]$. Let $v_p'$ be
the parent of $v_p$.
Note that $\ell(v_p')$ must be of the form $(\pi(y),u)$.
Now, let $(T',\ell)$ be the result of removing all subtrees rooted at some $v_p$.
Observe that $(T',\ell)$ is a partial $(\pi(y),s_i,\pi(y),s_{i+1})$-run for $\Abf$ in $\canmod|_{\pi(y)}$.
Let $\pathreq$ be the set of states $u$ such that there is a leaf labeled
$(\pi(y),u)$ and $u$ is not a final state.
 Since $\pi(y) = p \cdot r_0 C$, by Fact \ref{canmod2}, 
 $\canmod|_{\pi(y)}$ is isomorphic to the canonical model of
 $\tuple{\Tmc,C(a)}$, hence there is a partial run for $\Abf$ in the latter interpretation
 that has the same features.
Therefore we know that $(C,s_i,s_{i+1},\pathreq) \in \loops$ and in Step 5,
we replace $\Abf_{s_i,s_{i+1}}(y,y)$ by $\{
\ptest{\Abf_{u,F_k}}(y)\,{\mid}\,u\,{\in}\,\pathreq\,{\cap}\,S_k  \}$.
Moreover, for each $u \in \pathreq$, the subtree of $T$ rooted at the node
$v_p'$ with $\ell(v_p') = (\pi(y),u)$ (which became a leaf of $T'$), together
with the labeling $\ell$, is a full
$(\pi(y),u, F_k)$-run for  $\Abf$ in \canmod, hence
$\pi(y) \in \ptest{\Abf_{u,F_k}}^{\canmod}$.
%

The proof is similar when we have an  atom $\ptest{\Abf_{s_i,s_{i+1}}}(t)$
such that the witnessing sequence $u_{j_i}o_{j_i} \ldots u_{j_{i+1}} o_{j_{i+1}}$
for $\pi(y) \in \ptest{\Abf_{s_i,s_{i+1}}}^{\canmod}$
only involves objects from $\canmod|_{\pi(y)}$.
Then there must exist a
full $(\pi(y),s_i,\{s_{i+1}\})$-run $(T,\ell)$ for  $\Abf$ in $\canmod$
whose sequence of labels along the main path is given by
$(\pi(y), s_i)=(o_{j_i}, u_{j_i}),
(o_{j_i +1}, u_{j_i+1}), \ldots, (o_{j_{i+1}}, u_{j_{i+1}}) =(o_{j_{i+1}}, s_{i+1})$.
%
%
%
%
From the full run $(T,\ell)$ in $\canmod$, we obtain a partial run that allows us to
conclude the existence of some $(C,s_i,\{s_{i+1}\},\pathreq) \in \floops$, and
we can thus replace $\at$ by
by atoms $\{
\ptest{\Abf_{u,F_k}}(y)\,{\mid}\,u\,{\in}\,\pathreq\,{\cap}\,S_k\}$.
For all of these atoms,  we have
$\pi(y) \in \ptest{\Abf_{u,F_k}}^{\canmod}$.

The final choices to be made occur in Step~6, where we must
choose a conjunction of concept names $D$, roles $r, r_{1}, r_{2} \in \rni$, and states
such that conditions (a)-(d) are satisfied.
We set $r=r_0$ (recall that $\pi(y)= p \cdot r_0\, C$) and let $D$ be the conjunction
of all concept names $A$ such that $p \in A^{\canmod}$.
The definition of canonical models,
together with our normal form for \elhibot TBoxes,
yields
$\Tmc \models D \sqsubseteq \exists r.C$.
It remains to show that we can find $r_{1}, r_{2}$ and choices of states
such that conditions (a) -- (d) 
are verified.
Observe that all atoms of the form
 $\Abf_{u,U}(y,t)$, $\Abf_{u,U}(t,y)$ and $\ptest{\Abf_{u,U}}(y)$
 in $q$ were obtained by a decomposition in Step~4 (this may include original
 atoms that were decomposed into only one atom), or were added in Step~5
 as a result of some chosen $\pathreq$.
If we have an atom $\Abf_{u,U}(t_1,t_2)$ with
$y\in\{t_1,t_2\}$, it must have been added in Step~4.
We have seen that $(\pi(t_1),\pi(t_2)) \in \Abf_{u,U}^{\canmod}$, and
since the atom was not removed in Step~5, the witnessing sequence,
which we will denote $u_0'o_0'\cdots o_{l-1}'u_l'o_l'$,
lies outside the tree
$\canmod|_{\pi(y)}$: it may start or end
at $\pi(y)$, but it contains no other domain element from $\canmod|_{\pi(y)}$.
We also know that $l \geq 1$ (else again we would have removed the atom in Step 5),
and so the sequence must pass the parent $p$ of
$\pi(y)$. More specifically, if $\sigma(t_1)=y$ and $u \in S_i$,
then $o_0'=\pi(y)$, $u_0' = u$,
$o_1'=p$, and and $u_1' =
v$ for some $v \in S_i$ such that $(u,r_1^-,v) \in \delta_i$
and $\Tmc \models r_0 \ISA r_1$.
Note that this $v$ is such that $(p,\pi(t_2)) \in \Int[\canmod]{{\Abf_{v,U}}}$.
If $\sigma(t_2)=y$ and $u \in S_i$, then  $o_{l-1}'=p$, $o_l'=\pi(y)$, $u_{l-1}' = v'$,
and $u_l
= v''$ for some $v',v'' \in S_i$ such that $(v',r_2,v'') \in \delta_i$, $v'' \in U$, and $\Tmc
\models r_0 \ISA r_2$.
In this case, we have $(\pi(t_1),p) \in
{\Abf_{u,v'}}^{\canmod}$,
and $(p,\pi(t_2)) \in
{\Abf_{v',v''}}^{\canmod}$.
If $\sigma(t_1)=\sigma(t_2)=y$, then we pass by $p$ both at the beginning and end
of the witnessing sequence, and by choosing states $v,v'$ as above, we obtain
$(p,p) \in
{\Abf_{v,v'}}^{\canmod}$.
Finally, if we have an atom of the form $\ptest{\Abf_{u,U}}(y)$, then
we have seen previously that we must have
$\pi(y) \in\ptest{\Abf_{u,U}}^{\canmod}$. Moreover, since the
atom was either added in Step 4, but not dropped in Step 5, or
was added in Step 5, it must be the case that
the witnessing sequence $u_0'o_0'\cdots o_{l-1}'u_l'o_l'$
for  $\pi(y) \in\ptest{\Abf_{u,U}}^{\canmod}$ is such that
$o_0' = \pi(y)$, $u_0' = u$, $o_1' = p$,  and $u_1' =
v$ for some $v \in S_i$ such that $(u,r_1^-,v) \in \delta_i$ and
and $\Tmc \models r_0 \ISA r_1$.
Note that this $v$ is such that $p \in \ptest{\Abf_{v,U}}^{\canmod}$.
By selecting $r = r_0$, $r_1$, $r_2$ and  $v$, $v'$ in this way,
conditions (a) -- (d) 
are all verified.

Now let $q'$ be the query we obtain at the end of Step~8 when all non-deterministic choices
are made in the manner described above.
Note that  $\terms(q') \subseteq \terms(q)$.
We aim to find a match $\pi'$
for $q'$ which satisfies the conditions of the lemma.
Let  $\pi'$ be the mapping defined as follows:
\begin{itemize}
\item $\pi'(u)=\pi(u)$ for every $u \in \terms(q)$ with $\pi(u) \neq \pi(y)$
\item $\pi'(u)= p$ for every $u \in \terms(q)$ with $\pi(u)=\pi(y)$
\end{itemize}
Clearly,  $\pi'$ satisfies the conditions of the lemma. We only need to show
that $\pi'$ is a match.

To show that $\pi'$ is a match, first take some concept atom $B(u) \in q'$.
There are two possibilities. Either $B(u)$ appears in $q$
and $u \not \in \leaf$, or $B(u)$ was introduced in Step~7.
In the former case, we know that $\pi$ satisfies $B(u)$,
and since  $\pi'(u)=\pi(u)$ (since $u \not \in \leaf$), the same is true of $\pi'$.
In the latter case, we must have $u=y$ and $B \in  D$.
As $\pi'(u)=p$ and $D$ was chosen so that $p \in D^{\canmod}$,
 $\pi'$ satisfies $B(u)$.

Now consider some atom $\Abf_{s,F}(t,t') \in q'$ with  $y \not \in \{t,t'\}$.
Then $\Abf_{s,F}(t,t') \in q$.
As $\pi$ is a match for $q$ in $\Imc_{\Tmc,\Amc}$,
it must be the case that $(\pi(t),\pi(t')) \in \Abf_{s,F}^{\canmod}$.
Since $\pi'(t)=\pi(t)$ and $\pi'(t')=\pi(t')$, the same holds for $\pi'$,
and so the atom $\Abf_{s,F}(t,t')$ is satisfied by $\pi'$.
Similarly, for an atom $\ptest{\Abf_{s,F}}(t) \in q'$ with $t \neq y$,
we know that $\Abf_{s,F}(t,t') \in q$.
As $\pi$ is a match for $q$ in $\Imc_{\Tmc,\Amc}$,
it must be the case that $\pi(t) \in \ptest{\Abf_{s,F}}^{\canmod}$.
Since $\pi'(t)=\pi(t)$, the same holds for $\pi'$,
and so the atom $\ptest{\Abf_{s,F}}(t)$ is satisfied by $\pi'$.

Next consider some atom $\at \in q'$ of the form $\Abf_{u',U'}(t_1,t_2)$
or $\ptest{\Abf_{u',U'}}(t_1)$ with $y \in \{t_1,t_2\}$.
A straightforward examination of the procedure $\rewrite$
shows that there is an atom $\Abf_{s,F}(t',t'')$ or $\ptest{\Abf_{s,F}}(t')$
in $q$
which is replaced in Step~4 by the atoms
$\Abf_{s_{0},s_1}(\sigma(t'),y),$ $\Abf_{s_1,s_2}(y,y),$ $\ldots,
\Abf_{s_{n-2},s_{n-1}}(y,y),$ and either $ \Abf_{s_{n-1},s_n}(y,\sigma(t''))$
or $ \ptest{\Abf_{s_{n-1},s_n}}(y)$,
and one of the latter atoms is then replaced by the atom $\at$ in Step~7.
We distinguish four cases:\\[.15cm]
\noindent\textbf{Case 1}: $\at$ 
replaces $\Abf_{u,U}(y,t)$ with
$t \neq y$. Then $\at$ 
must have the form $\Abf_{v,U}(y,t)$,
where $v$ is the state which was chosen to ensure condition (b) in Step 6.
We recall that $v$ is such that
$(p,\pi(t)) \in \Abf_{v,U}^{\canmod}$.
Since $t \neq y$, we know that $t \not \in \leaf$, and so $\pi(t)=\pi'(t)$.
It follows that $(\pi'(y), \pi'(t))= (p,\pi'(t)) \in \Abf_{v,U}^{\canmod}$
and the atom $\at$ 
is satisfied by $\pi'$. \\[.15cm]
\noindent\textbf{Case 2}: $\at$ 
replaces $\Abf_{u,U}(t,y)$ with $t \neq y$.
Then $\at$ 
 must have the form $\Abf_{u,v}(t,y)$,
where $v$ is the state which was used in condition 6(c).
We showed earlier when examining condition 6(c) that $(\pi(t),p) \in
\Abf_{u,v}^{\canmod}$. Using the fact that $\pi'(t)=\pi(t)$ and
$\pi'(y)=p$, we can infer that
$(\pi'(t),\pi'(y)) \in \Abf_{u,v}^{\canmod}$, so $\pi'$
satisfies the atom $\at$. 
\\[.15cm]
\noindent\textbf{Case 3}: $\at$ 
 replaces $\Abf_{u,U}(y,y)$.
Then $\at$ 
 must have the form $\Abf_{v,v'}(y,y)$,
where $v$ is the state from 6(b) and $v'$ is the state from 6(c).
We have seen above that  $(p,p) \in
{\Abf_{v,v'}}^{\canmod}$. Since $\pi'(y) = p$,
$\pi'$ satisfies $\at$. 
\\[.15cm]
\noindent\textbf{Case 4}: $\at$ 
replaces $\ptest{\Abf_{u,U}}(y)$.
Then $\at$ 
 must have the form $\ptest{\Abf_{v,U}}(y)$,
where $v$ is the state which was chosen to ensure condition (d) in Step 6.
We recall that $v$ is such that
$p \in \ptest{\Abf_{v,U}}^{\canmod}$.
As $\pi'(y)=p$, $\pi'$ satisfies $\at$. 
\\[.2cm]
\indent
As we have shown that every atom in $q'$ is satisfied by the mapping $\pi'$,
it follows that $\pi'$ is a match for $q'$ in $\canmod$, which completes the proof.
\end{proof}

\medskip

\def\maxlev{\mu}

\noindent\textbf{Proposition \ref{evalatom-props}}.
\emph{$\evalatom$ is a sound and complete procedure for \nq\ evaluation over satisfiable \elhibot\ KBs.
It can be implemented so as to run in non-deterministic logarithmic space (resp. polynomial time) in the size of the ABox 
for \dlh\ (resp. \elhibot) KBs.}
\begin{proof}
Consider an \nnfa $(\Abf,s_0,F_0)$ and
a satisfiable \elhibot\
KB $\Kmc= \langle \Tmc, \Amc \rangle$ in normal form.
%
We prove soundness and completeness of $\evalatom$
by induction on the level of the automaton in the input \nnfa.
By Lemma \ref{lift-canmod} and Fact \ref{run-fact}, it suffices to show the following: \\ 

\noindent\textbf{Claim 1.} (Soundness and Completeness) 
\begin{enumerate}
\item There is a full $(a,s_0,b, s_f)$-run of $\Abf$ in $\canmod$ for some  $s_f \in F_0$
 if and only if there is an execution of \evalatom\ which returns \yes\
when given $(\Abf,s_0,F_0)$, $\tuple{\Tmc,\Amc}$, and $(a,b)$ as input, where $b \in \ainds(\Amc)$.
\item There is a full $(a,s_0,F_0)$
-run of $\Abf$ in $\canmod$ 
if and only if
there is an execution of \evalatom\ which returns \yes\
when given $(\Abf,s_0,F_0)$, $\tuple{\Tmc,\Amc}$, and $(a,\anon)$ as input.
\end{enumerate}

\noindent\textbf{Base Case of Claim 1.} Suppose that $s_0$ belongs to an automaton
$\alpha_j \in \Abf$
with $\level(j)=0$.

To show the first direction of statement 1,
suppose that there is an execution of \evalatom\ which returns \yes\
when given $(\Abf,s_0,F_0)$, $\tuple{\Tmc,\Amc}$, and $(a,b)$ as input,
As $b \in \ainds(\Amc)$, we know that in Step 3(b), case (iii) does not apply,
and because of the way $\loops$ is defined, we know that if
$(C,s, s',\pathreq)\in \loops$, then the set $\pathreq$ must be empty.
It follows that when \evalatom\ is called on the \nnfa\ $(\Abf,s_0,F_0)$,
there will be no recursive calls made to \evalatom.
Let $(c_0,s_0) (c_1, s_1) \ldots (c_m,s_m)$ be the sequence of
elements in $\current$ during the execution of \evalatom.
Note that $c_0=a$ and since \evalatom\ returned \yes, we also have $c_m=b$ and $s_m \in F_0$.
It thus suffices to show the following: \smallskip \\
\noindent\textbf{Claim 2}. For every $0 \leq i \leq m$,
there exists a full $(c_i,s_i,c_m,s_m)$-run of $\Abf$ on $\canmod$.
\\[.15cm]
\emph{Proof of claim}. The proof is by induction on $i$.
The base case is when $i=m$, in which case we can use the
partial run having a single
node labelled $(c_m,s_m)$.
Now suppose that the claim holds for all $k < i \leq m$,
and consider the case when $i=k$. We know that
the pair $(c_{k+1}, s_{k+1})$ must satisfy either condition (i) or (ii) of Step 3(b).
If (i) holds, then there exists a transition
$(s_{k}, \sigma, s_{k+1}) \in \delta_j$ with $\sigma \in \sigmaroles$
such that $(c_{k},c_{k+1}) \in \sigma^{\canmod}$.
From the induction hypothesis, we can find a full $(c_{k+1}, s_{k+1}, c_m,s_m)$-run
$(T, \ell)$ of $\Abf$ on $\canmod$. Let $(T', \ell')$ be the labelled tree obtained
by creating a new node labelled $(c_k,s_k)$ and making it the parent of the root of
$(T,\ell)$. It is easily verified that $(T',\ell')$ is defines a full
$(c_{k}, s_{k}, c_m,s_m)$-run of $\Abf$ on $\canmod$.
Next consider the case in which condition (ii) holds for the pair $(c_{k+1}, s_{k+1})$.
Then $c_k=c_{k+1}$ and there exists a tuple $(C,s_k, s_{k+1},\emptyset)$
in $\loops$ such that $c_k \in C^{\canmod}$.
By Definition \ref{jumpdef}, the presence of $(C,s_k, s_{k+1},\emptyset)$ in $\loops$
means that there is a partial run $(T_C,\ell_C)$ of $\Abf$
in the canonical model 
of $\tuple{\Tmc, \{C(e)\}}$
that satisfies the following conditions:
\begin{itemize}
\item the root of $T_C$ is labelled $(e,s_k)$
\item there is a unique leaf node $v$ 
with $\ell_C(v)=(e,s_{k+1})$
\end{itemize}
Since $c_k \in C^{\canmod}$, it follows from Fact \ref{canmod2}
that there is a homomorphism $h$ from
the canonical model of $\tuple{\Tmc, \{C(e)\}}$ to $\canmod$ with $h(e)=c_k$.
Let $(T_C', \ell_C')$ be obtained by replacing every label $(o,s)$ in $T_C$
by $(h(o),s)$. Using the fact that $h$ is a homomorphism, one can show that
$(T_C', \ell_C')$ defines a partial run of $\Abf$
in $\canmod$ that satisfies the following conditions:
\begin{itemize}
\item the root of $T_C'$ is labelled $(c_k,s_k)$
\item there is a unique leaf node $v$ 
with $\ell_C'(v)=(c_k,s_{k+1})$
\end{itemize}
Now let $(T,\ell)$ be the labelled tree obtained from $(T_C', \ell_C')$ by
replacing the unique leaf node labelled $(c_k,s_{k+1})$ by the
tree $(T_{k+1}, \ell_{k+1})$.
It follows from the properties of the component runs that $(T, \ell)$ is
a full $(c_{k},s_{k},c_m,s_m)$-run  of $\Abf$ on $\canmod$.
(\emph{end proof of Claim 2})


The first direction of statement 2 is proved similarly.
We let $(c_0,s_0) (c_1, s_1) \ldots (c_m,s_m)$ be the sequence of
elements in $\current$ during the execution of \evalatom.
If $c_n \in \ainds(\Amc)$, then we can use the same argument as above
to construct a full $(a,s_0,b,s_m)$-run of $\Abf$ in $\canmod$.
Consider next the case in which $c_m = \anon$. Then for $(c_m,s_m)$
to be selected, the conditions in 3(b)(iii) must have satisfied,
and so there exists some tuple $(C,s_{m-1}, F_0, \emptyset) \in \floops$
such that $c_{m-1} \in C^{\canmod}$.
It follows that there is a partial run $(T,\ell)$ for $\Abf$
on the canonical model $\Jmc$ of $\Tmc$ and $\{C(d)\}$ which consists of a single path
and whose sequence $\omega$ of labels from root to leaf
begins with $(c_{m-1},s_{m-1})$ and ends by $(o,s_f)$, for some $o \in \Delta^\Jmc$ and $s_f \in F_0$.
As $c_{m-1} \in C^{\canmod}$, it follows from Fact \ref{canmod2}
that there is a homomorphism $h$ from $\Jmc$ to $\canmod|_{c_{m-1}}$
with $h(d)=c_{m-1}$.
Let $\omega'$ be the sequence obtained by replacing each label $(o',s')$ in $\omega$
by $(h(o'), s')$. We then construct a labelled tree
consisting of a single path and whose sequence of labels (from root to leaf)
is $(c_0,s_0) (c_1, s_1) \ldots (c_{m-2},s_{m-2}) \omega'$. It can be verified that
this labelled tree defines a full $(a,s_0,F_0)$
-run of $\Abf$ in $\canmod$.

To prove the second direction of statement 1,
suppose that $(T,\ell)$ is a full $(a,s_0,b,s_f)$-run of $\Abf$ in $\canmod$,
where $s_f \in F_0$.
Note that because $\level(j)=0$, the tree $T$ consists of a single path.
Let $(o_0,s_0) \ldots (o_m, s_m)$ be the sequence of labels from root to leaf.
We have that 
$o_0=a$, $o_m=b$, and $s_m = s_f$.
Moreover, we know that for each $i\in\{1,\ldots,m\}$, there is a
transition $(s_{i-1},\sigma_i,s_i)\in\delta_j$ such that 
$\sigma_i\in\sigmaroles$ and $(o_{i-1},o_i)\in \sigma_i^{\canmod}$ (here we use the
fact that there are no transitions in $\alpha_j$ of type $\ptest{k}$).
We can suppose without loss of generality that there do not exist $0\leq  i, i'\leq m$
such that $i \neq i'$, $s_i=s_{i'}$, and $o_i=o_{i'}$ (indeed, this condition
can always be ensured by deleting a portion of the path).
Now let $i_1< ...< i_p$ be all of the indices $i$
such that $o_i \in \ainds(\Amc)$ (note that $i_1=0$ and $i_p=m$).
\smallskip \\
\noindent\textbf{Claim 3}.
For every $1 < l \leq p$, one of the following holds:
\begin{enumerate}
\item 
$(o_{i_{l -1}},o_{i_l})\in \sigma{_{i_l}}^{\canmod}$, or
\item $o_{i_l}= o_{i_{l-1}}$ and there exists a tuple $(C,s_{i_{l-1}}, s_{i_l}, \emptyset)\in \loops$ such that
$o_{i_{l-1}} \in C^{\canmod}$.
\end{enumerate}
\emph{Proof of claim}. If $i_{l} = i_{l -1} + 1$, then we immediately obtain $(o_{i_{l -1}},o_{i_l})\in \Int{\sigma{_{i_l}}}$.
If $i_{l} \neq i_{l -1} + 1$, then $o_{i_{l -1} + 1} \not \in \ainds(\Amc)$. We can infer from this that
all the objects $o_{i_{l -1} + 1}, \ldots, o_{i_{l } - 1}$
must be descendants of $o_{i_l}= o_{i_{l-1}}$ in $\canmod$.
Now let $C$ be the conjunction of all concept names $A$ such that $o_{i_{l-1}} \in A^{\canmod}$,
and let $\Jmc$ be the canonical model of $\Tmc$ and $\{C(d)\}$.
Then by Fact \ref{canmod2}, there must exist a
homomorphism $h$ from $\canmod$ to $\Jmc$ with $h(o_{i_{l-1}})=d$.
Consider the labelled tree $(T', \ell')$ obtained from $(T,\ell)$ by (1)
making the node labelled $(o_{i_{l-1}}, s_{i_{l-1}})$ the root and
the node labelled $(o_{i_l}, s_{i_l})$ the unique leaf, and
then (2) replacing every node label $(o',s')$ by $(h(o'),s')$.
Using the fact that $h$ is a homomorphism, one can show
that $(T',\ell')$ is a partial run for $\Abf$ in~$\Jmc$.
Moreover, since $o_{i_l}= o_{i_{l-1}}$ and $h(o_{i_{l-1}})=d$,
$(T', \ell')$ is a  $(d,s_{i_{l-1}}, d, s_{i_l})$-run of $\Abf$ on $\Jmc$.
Hence, by Definition \ref{jumpdef},
$(C,s_{i_{l-1}}, s_{i_l},\emptyset)\in \loops$.
(\emph{end proof of Claim 3})\smallskip\\
Now consider the sequence of pairs $(o_{i_1}, s_{i_1}) \ldots (o_{i_p},s_{i_p})$. Recall that
 $(o_{i_1}, s_{i_1})=(a,s_0)$, $(o_{i_p},s_{i_p})=(b,s_f)$, $s_f \in F_0$,
 and by construction $p \leq |\Amc| \times |S_j|$. 
 Therefore, using Claim 3, we can
show that the execution of $\evalatom$ which guesses this
sequence of elements in $\current$ during Step 3 will return \yes.

The proof of the second direction of statement 2 is broadly similar.
We let $(T,\ell)$ be a full $(a,s_0,F_0)$
-run of $\Abf$ in $\canmod$,
and we suppose w.l.o.g. that  $(T, \ell)$ is a shortest such run.
Let $(a,s_0)=(o_0,s_0) (o_1,s_1) \ldots (o_m, s_m)$
be the sequence of labels in $T$ from root to leaf. Note that we must have $s_m \in F_0$.
If $o_m \in \ainds(\Amc)$, then we can use
the same proof as above to show that there is a successful execution of \evalatom.
If $o_m \not \in \ainds(\Amc)$,
then let  $i_1 < \ldots < i_p$ be all those indices $i$
such that $o_i \in \ainds(\Amc)$. Note that $i_p < m$ since
$o_m \not \in \ainds(\Amc)$ and that $o_{i_{p} + 1}, \ldots, o_{m}$
must also belong to $\canmod|_{o_{i_p}}$.
We let $C$ be the conjunction of all concept names $A$ such that $o_{i_{p}} \in A^{\canmod}$,
and let $\Jmc$ be the canonical model of $\Tmc$ and $\{C(d)\}$.
Because of Fact \ref{canmod2}, we know that 
there exists a homomorphism $h$ from $\canmod$ to $\Jmc$ with $h(o_{i_{p}})=d$.
Consider the labelled tree $(T', \ell')$ obtained from $(T,\ell)$ by (1)
making the node labelled $(o_{i_p}, s_{i_p})$ the root and
then (2) replacing every node label $(o',s')$ by $(h(o'),s')$.
Using the fact that $h$ is a homomorphism, one can show
that $(T',\ell')$ is a  $(d,s_{i_p}, F_0)$ 
-run of $\Abf$ on $\Jmc$.
We thus have $(C,s_{i_{p}}, F_0, \emptyset)\in \floops$ by Definition \ref{jumpdef}.
Using this fact together with Claim 3,
we can show that the execution of $\evalatom$ which guesses
the sequence $(o_{i_1}, s_{i_1}) \ldots (o_{i_p},s_{i_p})$
of elements in $\current$ will return \yes.
This completes the proof of the base case.
 \medskip \\
\noindent\textbf{Induction Step of Claim 1.}
Assume that Claim 1 holds
whenever the initial $s_0$ and final states $F_0$
come from automaton having level at most $g$, and let us now
consider the case in which $s_0$ and $F_0$
belong to $\alpha_j$ whose level is $g+1$.

To show the first direction of statement 1,
suppose that there is an execution of \evalatom\ which returns \yes\
when given $(\Abf,s_0,F_0)$, $\tuple{\Tmc,\Amc}$, and $(a,b)$ as input.
Let $(c_0,s_0) (c_1, s_1) \ldots (c_m,s_m)$ be the sequence of
elements placed in $\current$ during the top-level execution of \evalatom.
Note that we must have $c_0=a$ and since \evalatom\ returned \yes,
we must also have $c_m=b$ and $s_m \in F_0$.
It thus suffices to show the following: \smallskip \\
\noindent\textbf{Claim 4}. For every $0 \leq i \leq m$,
there exists a full $(c_i,s_i,c_m,s_m)$-run of $\Abf$ on $\canmod$.
\\[.15cm]
\emph{Proof of claim}. The proof is by induction on $i$.
The base case is when $i=m$, in which case it suffices to consider the
tree having a single
node labelled $(c_m,s_m)$.
Now suppose that the claim holds for all $k < i \leq m$,
and consider the case when $i=k$. We know that
the pair $(c_{k+1}, s_{k+1})$ satisfies either condition (i) or (ii) of Step 3(b).
If (i) holds, then we can construct
a full $(c_{k}, s_{k}, c_m,s_m)$-run
of $\Abf$ on $\canmod$, in the same manner as in the base case.
Next suppose that it is condition (ii) that holds for the pair $(c_{k+1}, s_{k+1})$.
Then $c_k=c_{k+1}$ and there exists a tuple $(C,s_k, s_{k+1},\pathreq)$
in $\loops$ such that $c_k \in C^{\canmod}$ and
for every state $u \in \pathreq\cap S_l$, 
\begin{equation}
\evalatom((\Abf,u,F_l), \Kmc, (c_k,\anon)) = \yes
\end{equation}
By Definition \ref{jumpdef}, the presence of $(C,s_k, s_{k+1},\pathreq)$ in $\loops$
means that there is a partial run $(T_C,\ell_C)$ of $\Abf$
in the canonical model 
of $\tuple{\Tmc, \{C(e)\}}$
that satisfies the following conditions:
\begin{itemize}
\item the root of $T_C$ is labelled $(e,s_k)$
\item there is a leaf node $v$ 
with $\ell_C(v)=(e,s_{k+1})$
\item for every leaf node $v$ 
 with $\ell_C(v)=(o,s)\neq (e,s_{k+1})$, 
 either $s\in F_{j'}$ for some $j' > j$,
or $o=e$ and $s \in \pathreq$
\end{itemize}
Since $c_k \in C^{\canmod}$, it follows from Fact \ref{canmod2}
that there is a homomorphism $h$ from
the canonical model of $\tuple{\Tmc, \{C(e)\}}$ to $\canmod$ with $h(e)=c_k$.
Let $(T_C', \ell_C')$ be obtained by replacing every label $(o,s)$ in $T_C$
by $(h(o),s)$. Using the fact that $h$ is a homomorphism, one can show that
$(T_C', \ell_C')$ defines a partial run of $\Abf$
in $\canmod$ that satisfies the following conditions:
\begin{itemize}
\item the root of $T_C'$ is labelled $(c_k,s_k)$
\item there is a leaf node $v$ 
with $\ell_C'(v)=(c_k,s_{k+1})$
\item for every leaf node $v$ 
 with $\ell_C'(v)=(o,s)\neq (c_k,s_{k+1})$, 
 either $s\in F_{j'}$ for some $j' > j$,
or $o=c_k$ and $s \in \pathreq$
\end{itemize}
Next, consider the states in $\pathreq$. By Equation 6 and the induction hypothesis for Claim 1,
we can find for every $u \in \pathreq\cap S_l$ a full $(c_k,u,F_l)$-run
$(T_u, \ell_u)$ of $\Abf$ on $\canmod$.
Finally, by applying the induction hypothesis of the present claim, and utilizing the fact that $c_k=c_{k+1}$,
we can find a full $(c_{k},s_{k+1},c_m,s_m)$-run $(T_{k+1}, \ell_{k+1})$ of $\Abf$ on $\canmod$.
It then suffices to plug these different runs together. Specifically, we let
$(T,\ell)$ be the labelled tree obtained from $(T_C', \ell_C')$ by:
\begin{itemize}
\item replacing the (unique) leaf node labelled $(c_k,s_{k+1})$ by the tree $(T_{k+1}, \ell_{k+1})$
\item replacing each leaf node labelled $(c_k,u)$ with $u \in \pathreq \cap S_l$
by the tree $(T_u, \ell_u)$
\end{itemize}
It follows from the properties of the component runs that $(T, \ell)$ is
a full $(c_{k},s_{k},c_m,s_m)$-run  of $\Abf$ on $\canmod$.
(\emph{end proof of Claim 4})

The proof of the first direction of statement 1 is similar.
Indeed, if the final element $(c_m,s_m)$ in $\current$ is such  that $c_m \in \ainds(\Amc)$,
then by Claim 4, there is a full $(a,s_{0}, c_m,s_m)$-run  of $\Abf$ on $\canmod$,
which is also a full $(a,s_{0},F_0)$-run.
If $c_m = \anon$, then proof is generally similar, except that instead of Claim 4, 
we use the following claim:\smallskip \\
\noindent\textbf{Claim 5}. For every $0 \leq i < m$,
there exists a full $(c_i,s_i,F_0)$-run of $\Abf$ on $\canmod$.
\\[.15cm]
\emph{Proof of claim}. The base case is $i=m-1$.
As $c_m=\anon$, the condition in Step 3(b)(iii) must have been satisfied.
Thus, there exists a tuple $(C,s_{m-1},F_0,\pathreq)\in \floops$
such that $c _{m-1} \in C^{\canmod}$ and for every $u \in \pathreq \cap S_l$,
$$\evalatom((\Abf,u,F_l), \Kmc, (c_{m-1},\anon)) = \yes$$
As $c _{m-1} \in C^{\canmod}$ and $(C,s_{m-1},F_0,\pathreq)\in \floops$,
we can use the same reasoning as in the proof of Claim 4
to obtain a partial run $(T_C, \ell_C)$ of $\Abf$
in $\canmod$ that satisfies the following conditions:
\begin{itemize}
\item the root of $T_C$ is labelled $(c_{m-1},s_{m-1})$
\item there is a leaf node $v$ 
with $\ell_C(v)=(o,s_f)$ with $s_f \in F_0$
\item for every leaf node $v$ 
 with $\ell_C(v)=(o,s)$ with $s \not \in F_0$, 
either $s\in F_{j'}$ for some $j' > j$,
or $o=c_{m-1}$ and $s \in \pathreq$
\end{itemize}
We can also find, for each $u \in \pathreq\cap S_l$, a full $(c_k,u, F_l)$-run
$(T_u, \ell_u)$ of $\Abf$ on $\canmod$.
By replacing each leaf node in $(T_C, \ell_C)$ that is labelled $(c_{m-1},u)$ with
$u \in \pathreq$ by the tree
$(T_u, \ell_u)$, we obtain
a full $(c_{m-1},s_{m-1}, F_0)$-run of $\Abf$ on $\canmod$.
This establishes the base case.
The induction step then proceeds exactly as
in the proof of Claim 4. (\emph{end proof of Claim 5})\\

To show the second direction of statement 1, let us suppose that
$(T, \ell)$ is a full $(a,s_0,b, s_f)$-run of $\Abf$ in $\canmod$, where $s_f \in F_0$.
Let $v_0,Ê\ldots, v_m$ be the sequence of nodes in $T$ that begins with the root node
labelled $(a,s_0)$ and ends with the unique leaf node labelled $(b,s_f)$.
We will use $(o_l, s_l)$ for the label of node $n_l$.
We can assume without loss of generality that $(T, \ell)$ is minimal in
the sense that there do not exist distinct positions $0\leq  l, l' \leq m$
such that $n_l$ and $n_{l'}$ has the same label. 
Let $i_1< ...< i_p$ be all of the indices $i$
such that $o_i \in \ainds(\Amc)$. Note that $i_1=0$ and $i_p=m$ since
$i_0=a\in \ainds(\Amc)$ and $i_m=b \in \ainds(\Amc)$.
\smallskip \\
\noindent\textbf{Claim 6}.
For every $1 < l \leq p$, one of the following holds:
\begin{enumerate}
\item 
there is some $(s_{i_{l -1}},\sigma,s_{i_l}) \in \delta_j$ with $\sigma \in \sigmaroles$
such that $(o_{i_{l -1}},o_{i_l})\in \sigma{_{i_l}}^{\canmod}$
\item $o_{i_l}= o_{i_{l-1}}$ and there exists a tuple $(C,s_{i_{l-1}}, s_{i_l}, \pathreq)\in \loops$ such that
$o_{i_l-1} \in C^{\canmod}$ and for every $u \in \pathreq \cap S_{k}$,
$$\evalatom((\Abf,u,F_k), \Kmc, (o_{i_{l-1}},\anon)) = \yes$$
\end{enumerate}
\emph{Proof of claim}. It follows from Definition \ref{rundef} that there are three possibilities:
  \setlist{leftmargin=*,labelindent=.15cm}
\begin{itemize}
\item[(a)] the node $n_{i_{l-1}}$ has $n_{i_l}$ as its unique child
\item[(b)] the node $n_{i_{l-1}}$ has $n_{i_l}$ as one of its two children
\item[(c)] the node $n_{i_l}$ is a descendant, but not a child, of $n_{i_{l-1}}$
\end{itemize}
First consider case (a). Then by Definition \ref{rundef}, there must
exist a transition $(s_{i_{l-1}},\sigma,s_{i_l}) \in \delta_j$ such that $\sigma \in \sigmaroles$ and
$(o_{i_{l-1}},o_{i_l}) \in \sigma^{\canmod}$, so the first statement of the claim is satisfied.
If (b) holds, then it follows from Definition \ref{rundef} that $(s_{i_{l-1}},\ptest{k},s_{i_l}) \in \delta_j$,
$o_{i_l}= o_{i_{l-1}}$, and the other child of $n_{i_{l-1}}$ has label $(o_{i_{l-1}}, u)$,
with $u$ the initial state of $\alpha_k$. Since $(s_{i_{l-1}},\ptest{k},s_{i_l}) \in \delta_j$,
we have $(\top,s_{i_{l-1}}, s_{i_l}, \{u\})\in \loops$.
The labelled subtree of $(T,\ell)$ rooted at
$(o_{i_{l-1}}, u)$ is a full $(o_{i_{l-1}}, u, F_k)$-run of $\Abf$ in $\canmod$, and since
$k > j$, we have $\level(k) \leq g$.
The induction hypothesis for Claim 1 is thus applicable
and yields
$$\evalatom((\Abf,u,F_k), \Kmc, (o_{i_{l-1}},\anon)) = \yes$$
We have thus shown that the second statement of the claim holds.
Finally, let us consider case (c). Since $n_{i_l}$ is not a child of $n_{i_{l-1}}$,
we must have $o_{i_{l -1} + 1} \not \in \ainds(\Amc)$. We can infer from this that
all the objects $o_{i_{l -1} + 1}, \ldots, o_{i_{l } - 1}$
must be descendants of $o_{i_{l-1}}$ in $\canmod$, and that $o_{i_{l-1}}=o_{i_l}$.
Consider the partial run $(T',\ell')$ obtained from $(T, \ell)$ by
\begin{enumerate}
\item making $n_{i_{l-1}}$ the new root node
\item for every descendant $v$ of $n_{i_{l-1}}$ such that $\ell(v)=(o_{i_{l-1}},s)$ and such that
there is no $v'$ (different from $v$ and $n_{i_{l-1}}$) with $\ell(v')=(o_{i_{l-1}},s')$ and occurring along the path from $n_{i_{l-1}}$ to $v$,
make $v$ a leaf node by dropping all children of $v$
\end{enumerate}
Note that since $(T, \ell)$ is a full run, there are three types of leaf nodes in $(T',\ell')$:
\begin{itemize}
\item the node $n_{i_l}$ (on the main path) with label $(o_{i_{l-1}}, s_{i_l})$
\item nodes with labels of the form $(o_{i_{l-1}},s)$
\item nodes with labels of the form $(o,s_f)$ with $s_f \in F_k$ for some $k > j$
\end{itemize}
Moreover, by construction, every label $(o,s)$ appearing in $(T',\ell')$ is such that $o \in \canmod|_{o_{i_{l-1}}}$.
Now let $C$ be the conjunction of all concept names $A$ such that $o_{i_{l-1}} \in A^{\canmod}$,
and let $\Jmc$ be the canonical model of $\Tmc$ and $\{C(d)\}$.
Then it follows from Fact \ref{canmod2} there exists a
homomorphism $h$ from $\canmod$ to $\Jmc$ with $h(o_{i_{l-1}})=d$.
Consider the labelled tree $(T'', \ell'')$ obtained from $(T',\ell')$ by replacing
every node label $(o,s)$ by $(h(o),s)$.
Using the fact that $h$ is a homomorphism with $h(o_{i_{l-1}})=d$ and the
above description of the leaf nodes in $(T',\ell')$,
one can show that $(T'', \ell'')$
is a partial run of $\Abf$ on $\Jmc$ that satisfies:
\begin{itemize}
\item the root of $T''$ is labelled $(d, s_{i_{l-1}})$
\item there is a leaf node labelled $(d, s_{i_{l}})$
\item for every leaf node $v$ with $\ell(v) = (o, s) \neq (d, s_{i_{l}})$, either
$s \in F_k$ for some $k>j$, or $o=d$
\end{itemize}
Thus, by Definition \ref{jumpdef}, the tuple $(C, s_{i_{l-1}}, s_{i_{l}},\pathreq)$ belongs to $\loops$,
if we let $\pathreq$ contain all those states $u$ such that there is a leaf node $(d,u)$ with $u \neq s_{i_{l}}$.
To complete the proof, consider some $u \in \pathreq \cap S_{k}$. Then there is a leaf node of $(T'', \ell'')$
with label $(d,u)$, hence a leaf node of $(T', \ell')$ with label $(o_{i_{l-1}},u)$,
which also appears (perhaps not as a leaf) in the original run
$(T,\ell)$. Moreover, since $(T,\ell)$ is a full run, by taking the subtree of $(T,\ell)$ rooted at this node,
we obtain a full $(o_{i_{l-1}},u)$-run of $\Abf$ in $\canmod$. As $u$ belongs to $\alpha_k$, with $k > j$,
the induction hypothesis for Claim 1 is applicable, and yields
$$\evalatom((\Abf,u,F_k), \Kmc, (o_{i_{l-1}},\anon)) = \yes$$
We have thus shown that the second statement of the claim holds.
(\emph{end proof of Claim 6})\\[.15cm]
It follows from Claim 6 and the fact that $p \leq |\Amc| \times |S_j|$
that there is an execution of $\evalatom$ which guesses
the sequence $(o_{i_1}, s_{i_1}) \ldots (o_{i_p},s_{i_p})$
of elements in $\current$ that returns \yes.

The second direction of statement 2 can be proven similarly.
We suppose that
$(T, \ell)$ is a full $(a,s_0,F_0)$-run of $\Abf$ in $\canmod$, and let
$v_0,Ê\ldots, v_m$ be the sequence of nodes in $T$ that begins with the root node
labelled $(a,s_0)$ and ends with the unique leaf node labelled $(o,s)$ with $s \in S_j$.
We use $(o_l, s_l)$ for the label of node $n_l$. Note that since $(T, \ell)$ is a full
$(a,s_0,F_0)$-run, we have $s_m \in F_0$.
We define the sequence $i_1 < \ldots < i_p$
and make the same minimality assumption on $(T, \ell)$ to
ensure that $p \leq |\Amc| \times |S_j|$.
If $o_m \in \ainds(\Amc)$, we can follow the above proof for statement 1 exactly.
If $o_m \not \in \ainds(\Amc)$, then we must additionally show that:
\smallskip \\
\noindent\textbf{Claim 7}. There is a tuple $(C,s_{i_p}, F_0,\pathreq)\in\floops$ such that
$o_{i_p} \in C^{\canmod}$ and for every $u \in \pathreq \cap S_{k}$,
$$\evalatom((\Abf,u,F_k), \Kmc, (o_{i_{l-1}},\anon)) = \yes$$
\emph{Proof of claim}.
Let $(T',\ell')$ be the partial run obtained from the run $(T, \ell)$ by
\begin{enumerate}
\item making $n_{i_{p}}$ the new root node
\item for every descendant $v$ of $n_{i_{p}}$ such that $\ell(v)=(o_{i_{p}},s)$ and such that
there is no $v'$ (different from $v$ and $n_{i_{p}}$) with $\ell(v')=(o_{i_{p}},s')$ and occurring along the path from $n_{i_{p}}$ to $v$, make $v$ a leaf node by dropping all of its children
\end{enumerate}
Note that since $(T, \ell)$ is a full run, there are three types of leaf nodes in $(T',\ell')$:
\begin{itemize}
\item the node $n_{m}$ (on the main path) with label $(o_{m}, s_{m})$
\item nodes with labels of the form $(o_{i_{p}},s)$
\item nodes with labels of the form $(o,s_f)$ with $s_f \in F_k$ for some $k > j$
\end{itemize}
Moreover, by construction, every label $(o,s)$ appearing in $(T',\ell')$ is such that $o \in \canmod|_{o_{i_{p}}}$.
Now let $C$ be the conjunction of all concept names $A$ such that $o_{i_{p}} \in A^{\canmod}$,
and let $\Jmc$ be the canonical model of $\Tmc$ and $\{C(d)\}$.
By Fact \ref{canmod2}, there exists a
homomorphism $h$ from $\canmod$ to $\Jmc$ with $h(o_{i_{p}})=d$.
Consider the labelled tree $(T'', \ell'')$ obtained from $(T',\ell')$ by replacing
every node label $(o,s)$ by $(h(o),s)$.
Using the fact that $h$ is a homomorphism with $h(o_{i_{p}})=d$ and the
above description of the leaf nodes in $(T',\ell')$,
one can show that $(T'', \ell'')$
is a partial run of $\Abf$ on $\Jmc$ that satisfies:
\begin{itemize}
\item the root of $T''$ is labelled $(d, s_{i_{p}})$
\item there is a leaf node labelled $(d, s_{m})$
\item for every leaf node $v$ with $\ell(v) = (o, s) \neq (d, s_m)$, either
$s \in F_k$ for some $k>j$, or $o=d$
\end{itemize}
Then since $s_m \in F_0$, by Definition \ref{jumpdef}, the tuple $(C, s_{i_{p}},F_0,\pathreq)$ belongs to $\floops$,
for the $\pathreq$ containing exactly those states $u$ such that there is a leaf node $(d,u)$ with $u \neq s_m$.
To complete the proof, consider some $u \in \pathreq \cap S_{k}$. Then there is a leaf node of $(T'', \ell'')$
with label $(d,u)$, and hence a leaf node of $(T', \ell')$ with label $(o_{i_{p}},u)$.
By the construction of $(T', \ell')$, there is a node labelled $(o_{i_{p}},u)$ in the original run
$(T,\ell)$, and since $(T,\ell)$ is a full run,
there is a full $(o_{i_{p}},u, F_k)$-run of $\Abf$ in $\canmod$.
Applying the induction hypothesis for Claim 1, we obtain
$$\evalatom((\Abf,u,F_k), \Kmc, (o_{i_p},\anon)) = \yes$$
(\emph{end proof of Claim 7})\\[.15cm]
Using Claims 6 and 7 and the fact that $p\leq |\Amc| \times |S_j|$,
one can show that there is an execution of $\evalatom$ which guesses
the sequence $(o_{i_1}, s_{i_1}) \ldots (o_{i_p},s_{i_p}) (\anon, s_m)$
of elements in $\current$ that returns \yes. (\emph{end proof of Claim 1})
\bigskip

We now show the second part of the proposition,
namely that \evalatom\ runs in non-deterministic logarithmic space in the size
of $\Amc$. First observe that when the input \nnfa\ is $(\Abf, s_0, F_0)$,
the depth of the recursion cannot exceed $\level(s_0)$. The latter number
is bounded above by the number of automata in $\Abf$, and hence is independent
of the size of $\Amc$. Next we note that for each call of \evalatom, there are only three
pieces of information
that must be stored and whose size depends on $|\Amc|$:
the individual $c$ in $\current$, the ``next" individual $d$,
and the value of the counter $\counter$. Each of these pieces of information
can be stored using logarithmic space in $|\Amc|$.
We then remark that testing whether a tuple belongs to $\loops$ or $\floops$
can be done in constant time in $|\Amc|$, since it is ABox-independent.
The only checks that involve $\Amc$ are testing whether $(c,d) \in \sigma^{\canmod}$ in Step 3(b)(i)
and testing whether $c \in C^{\canmod}$ in Steps 3(b)(ii) and 3(b)(iii).
If $\Tmc$ is formulated in \elhibot, then both checks can be done in polynomial time in $|\Amc|$,
and for \dlh, this can be improved to non-deterministic logarithmic space in $|\Amc|$.
This follows from the known \ptime\ (resp.\ \nlspace) data complexity of instance checking in these logics~\cite{HuMS05,CDLLR07}.
Putting this all together, we have that \evalatom\ runs in polynomial time in $|\Amc|$ for \elhibot\ KBs,
and in non-deterministic logarithmic space in the size
of $\Amc$ for \dlh\ KBs.
\end{proof}
